\documentclass[12pt]{article}
\usepackage{times,amssymb,amsmath,amsfonts,eurosym,geometry,ulem,graphicx,caption,color,setspace,sectsty,comment,footmisc,caption,pdflscape,array,threeparttable,tikz,subcaption,caption,natbib,tabularx} 
\PassOptionsToPackage{draft}{graphicx}
\newsavebox{\measurebox}
\usepackage[scale=2]{ccicons}
\usepackage{pgfplots}
\usepgfplotslibrary{dateplot}
\usetikzlibrary{intersections,shapes.multipart}
\usepackage{pifont}
\usetikzlibrary{tikzmark}
\usetikzlibrary{shapes,arrows,backgrounds,positioning}
\usetikzlibrary{intersections,shapes.multipart}
\usepackage{tikz}
\usetikzlibrary{fit,calc}

\usetikzlibrary{intersections,shapes.multipart}
\usepackage{pifont}
\usetikzlibrary{tikzmark}
\usetikzlibrary{shapes,arrows,backgrounds,positioning}
\usetikzlibrary{intersections,shapes.multipart}
\usepackage{pifont}
\usepackage{tikz}
\usetikzlibrary{fit,calc}
\usepackage{threeparttable}
\usepackage{pdflscape}
\usepackage{afterpage}
\usepackage{capt-of}

\usepackage{footmisc}
\usepackage[colorlinks=true,
citecolor=blue,
linkcolor=blue,
anchorcolor=blue,
urlcolor=blue]{hyperref}

\usepackage{bibunits}
\defaultbibliographystyle{apalike2}
\defaultbibliography{literature}

\usepackage{tikz,pgfplots}
\usepackage{amsmath,amsthm}
\newtheorem{assumption}{Assumption}

\makeatletter
\newtheorem*{assumption*}{Assumption}
\makeatother
\newtheorem{lemma}{Lemma}

\newtheorem{proposition}{Proposition}

\theoremstyle{definition}

\newtheorem{example}{Example}

\usepackage[titletoc,toc,title,page,header]{appendix}
\usepackage{minitoc}

\noptcrule
\usepackage{caption}

\title{Causal Inference for Experiments with Latent Outcomes: Key Results and Their Implications for Design and Analysis}

\date{\today}

\author{Jiawei Fu\footnote{Assistant Professor, Duke University (\url{jiawei.fu@duke.edu})} \quad Donald P. Green\footnote{Burgess Professor of Political Science, Columbia University (\url{dpg2110@columbia.edu}).  \\ 
We thank David Broockman, Alex Coppock, Macartan Humphreys, Josh Kalla, Patrick Liu, Ryan Moore, Libby Jenke, Justin Grimmer, Jacob Montgomery, Kosuke Imai, Teppei Yamamoto, and participants of the UC San Diego Political Methodology and American Politics Speaker Series for their helpful comments on earlier drafts. We welcome feedback on this draft. A preliminary R package can be found at \href{https://github.com/Jiawei-Fu/LatentOutcomes}{Github}.
}}
\pgfplotsset{compat=1.18} 
\begin{document}

\maketitle
\singlespacing



\begin{abstract}
How should researchers analyze randomized experiments in which the main outcome is latent and measured in multiple ways but each measure contains some degree of error? We first identify a critical study-specific noncomparability problem in existing methods for handling multiple measurements, which often rely on strong modeling assumptions or arbitrary standardization. Such approaches render the resulting estimands noncomparable across studies. To address the problem, we describe design-based approaches that enable researchers to identify causal parameters of interest, suggest ways that experimental designs can be augmented so as to make assumptions more credible, and discuss empirical tests of key assumptions. We show that when experimental researchers invest appropriately in multiple outcome measures, an optimally weighted scaled index of these measures enables researchers to obtain efficient and interpretable estimates of causal parameters by applying standard regression. An empirical application illustrates the gains in precision and robustness that multiple outcome measures can provide.

\noindent 
\vspace{.1in}

\end{abstract}

\thispagestyle{empty}

\clearpage
\doparttoc 
\faketableofcontents 

\setcounter{page}{1}

\doublespacing
\section{Introduction}





\begin{bibunit}
Social scientists often seek to estimate the average causal effect of interventions and increasingly rely on experimental designs to do so convincingly. The statistical literature on experimental design and analysis has grown markedly in recent years.
However, one topic of special concern to social scientists has largely escaped attention, even from otherwise comprehensive textbooks: imperfect measurement of experimental outcomes.  \citet{angrist2009mostly}, \citet{gerber2012field}, and \citet{imbens2015causal} offer no sustained formal treatment of outcome measurement or how to analyze experiments in which outcomes are measured in more than one way. The models that these textbooks present implicitly assume that the observed outcome is the true underlying potential outcomes of interest.

In the social sciences, however, there is often slippage between the outcomes of interest and the proxy outcome measures at hand.  Constructs such as economic inequality, press freedom, political violence, corruption, and post-materialism are just a few examples of latent variables that are thought to be measured imperfectly, despite the sustained efforts of researchers. To illustrate how frequently this issue arises,  SI \ref{si:summarytab}
lists 14 representative articles 
published in the \emph{American Political Science Review} during the past five years.  
Each article uses multiple measures to gauge a latent outcome, but statistical practice varies widely: regression using additive indices constructed from standardized measures, regression using indices based on indices derived from principal components analysis or from inverse covariance algorithms, or nonlinear estimators rooted in item response theory. Given multiple imperfect measurements, how should researchers identify and estimate average treatment effects on the latent outcome?

Our paper builds on recent attempts to formalize the identification and estimation challenges that may arise when latent outcomes are measured with error. Like \citet{stoetzer2022causal}, we use potential outcomes notation to unify our discussion of experimental design and outcome measurement. We make four contributions, each of which has important implications for research design and estimation. 

First, we identify and address a critical \textit{study-specific noncomparability problem} in existing methods for handling multiple measurements: principal components analysis (PCA), inverse covariance weighting (ICW) \citep{anderson2008multiple}, item response theory model (IRT) \citep{stoetzer2022causal}, and inverse regression analysis (IRA) (\citealt{zhang2025inverse}). The estimands in these approaches are often distorted by the estimation techniques themselves, such that even when researchers aim to identify the average treatment effects on the \textit{same} latent outcome across different studies, the resulting estimates are \textit{not comparable}. This drawback poses a serious obstacle to any sustained research program.

Consider, for example, a scenario in which researchers design a large experiment to estimate the treatment effect on political attitudes using several survey-based measures. With any of the aforementioned methods, they would obtain an estimate. Meanwhile, another research team might investigate the same question using slightly different measures and obtain another estimate. However, these results would not be comparable and would not, in fact, target the same latent construct of political attitude. The key reason is that a latent variable has no intrinsic metric -- its meaning is entirely determined by the measurements used to define it. Widely used methods arbitrarily standardize and combine these measurements, which in turn alters the interpretation of the latent outcome depending on the scale of the observed variables.

Building on a largely overlooked literature dating back to at least \citet{costner1971utilizing} and \citet{alwin1974causal}, we propose an efficient design-based method for identifying latent treatment effects without standardizing the latent outcome in ways that render its units of measurement study-specific. As a result, estimated effects are more directly comparable across studies.

Second, adapting ideas from prior work by \citet{bagozzi1977structural}, \citet{sorbom1981structural}, and \citet{kano2001structural} to a potential outcomes framework, we propose an optimally \textit{weighted scaled index} (WSI) estimator for the experimental intervention’s average treatment effect on a latent outcome. As explained in section \ref{sec:frame}, it involves two key steps: (1) scaling the outcomes to address the study-specific noncomparability problem and (2) optimally weighting the scaled outcomes. 
As demonstrated through a series of simulations, the WSI achieves higher statistical power than existing alternatives (PCA, ICW, IRT, IRA). More importantly, under optimal weighting, it attains similar efficiency to traditional structural equation modeling (SEM) estimators, which typically use maximum likelihood estimation based on strong normality assumptions. 
WSI also facilitates visualization of estimated causal effects and randomization inference for hypothesis testing. 

Third, we show analytically and through simulation how multiple measures of a latent outcome can improve the precision with which this average treatment effect is estimated in section \ref{sec:var}. When conducting experiments, researchers often face a key design dilemma: given additional resources, should they increase the sample size or collect more measurements per unit? Our framework provides a useful decision rule, showing that the optimal choice depends on the reliability of the outcome measures. When reliability is low, the marginal gain from collecting an additional measurement is relatively high relative to the gains from increasing the number of subjects. In contrast, when reliability is high, the marginal gain from additional measurements diminishes, and increasing the sample size becomes preferable. These results highlight an underappreciated practical trade-off between gathering more observations and gathering additional outcome measures.  

Fourth, although our method (like any method) rests on assumptions that are not directly testable, the credibility of these assumptions hinges on design choices. This stands in contrast to existing approaches, which typically separate the design and estimation stages. For instance, we show how certain data collection strategies can help satisfy what would otherwise be unrealistic measurement assumptions.\footnote{Our approach to data collection also reduces reliance on constant treatment effects among all subjects, such as the hierarchical Item Response Theory (IRT) model proposed by \citet{stoetzer2022causal}. We offer design recommendations that allow the analyst to stay within a simple  modeling framework that accommodates heterogeneous treatment effects.}  
We are by no means the first to note that ``invalid'' measurements threaten unbiased inference or that redundant outcome measurements can improve precision; our contribution is to bring a coherent analytic framework to the discussion of experimental research design and outcome measurement so as to clarify the value of research designs that allocate resources to the collection of multiple measurements of a latent outcome. We also show that the latent outcome models may be viewed as nested alternatives to seemingly unrelated regression (SUR), allowing researchers to assess empirically whether the constraints imposed by the latent outcome model are consistent with the data.  If not, researchers can fall back on the more agnostic SUR approach in which each individual outcome measure is regressed on the treatment.

This paper is structured as follows. We begin by presenting a potential outcomes model that defines the target of inference and allows for some degree of slippage between the latent outcome we seek to measure and the observed measures at our disposal. Next, we show formally the conditions under which the average treatment effect on the posited latent outcome is identified.  When latent outcomes are measured with random error, this average treatment effect may be estimated consistently, but precision is enhanced when the researcher gathers multiple outcome measures. When measurement errors are systematic -- especially when mismeasurement operates differently for subjects assigned to treatment or control -- methods that are typically used to analyze experiments may no longer yield unbiased estimates.  We show how identification problems may be diagnosed empirically and addressed by diversifying the portfolio of measures and adopting an estimator that leverages the additional measures in a theoretically informed way.  A simulated example illustrates the mechanics of this approach.  In order to show its relevance to applied empirical work, we reanalyze the \citet{kalla2020reducing} experiment, which features a diverse array of outcome measures that improve precision and facilitate robustness checks.

\section{Framework}\label{sec:frame}

Consider a sample of $n$ units drawn from a superpopulation.\footnote{The finite population result is similar to the superpopulation case, except that variance estimation is conservative, in keeping with the standard causal inference literature \citep{imbens2015causal}.} Let $Z_i$ denote the treatment assignment for unit $i$. Generalizing the potential outcomes framework, we assume two potential latent variables: $\eta_i^1$ is the latent outcome that would be realized for subject $i$ if the treatment were administered, and $\eta_i^0$ is the latent outcome that would be realized for subject $i$ in the absence of treatment. We assume the stable unit treatment value assumption, which implies that  potential outcomes remain stable regardless of which subjects receive treatment \citep{rubin1980randomization}.
These latent variables are not directly observed; think of them instead as abstract concepts, such as authoritarianism or corruption, that may be measured with some degree of error. Suppose that there are $J$ observed measures of the latent outcome for each subject. Let the $j^{th}$ outcome measure for unit $i$ be $Y_{ij}= \lambda_j \eta_i + \epsilon_{ij} = \lambda_j[Z_i\eta^1_{i} + (1-Z_i)\eta^0_{i}] + \epsilon_{ij}$, where $\epsilon_{ij}$ is the measurement error with mean zero $\mathbb{E}[\epsilon_{ij}]=0$.\footnote{This assumption can be relaxed to any constant since  outcomes can be mean-centered.} More generally, we can allow the variance of the measurement errors to differ between the treatment and control groups by defining two separate measurement error terms $\epsilon_{ij}=Z_i\epsilon^1_{ij} + (1-Z_i)\epsilon^0_{ij}$, 
with $\mathbb{E}[\epsilon^1_{ij}]=\mathbb{E}[\epsilon^0_{ij}]=0$. In the superpopulation framework, $\{Z_i,\eta^1_i,\eta^0_i, \epsilon^1_{ij},\epsilon^0_{ij}\}_{i=1}^n$ is assumed to be an i.i.d. sample from a population distribution.
Figure \ref{fig:dgp} illustrates the relationship among the observed and unobserved variables.
The observed variables are depicted inside squares, and unobserved variables are depicted inside circles. The latent outcome $\eta_i$ can be represented as $\eta_i= \mathbb{E}\eta_i^0 + \tau Z_i + \zeta_i$, where $\zeta_i$ is the idiosyncratic disturbance: $\zeta_i=\eta_i^0 - \mathbb{E}\eta^0_i+Z_i[(\eta_i^1-\mathbb{E}\eta^1_i)-(\eta_i^0-\mathbb{E}\eta^0_i)]$. This term has mean zero by construction: $\mathbb{E}\zeta_i=0$.

We are interested in the causal effect of treatment on the latent variable, which \citet{stoetzer2022causal} call the latent treatment effect (LTE). Define the individual-level LTE as $\tau_i=\eta^1_{i}-\eta^0_{i}$. Because the individual-level LTE is unidentified, we focus on the average latent treatment effect (ALTE) defined as $\mathbb{E}[\eta^1_{i}-\eta^0_{i}]$.


\begin{figure}[!h]
    \centering
    \includegraphics[width=0.8\linewidth]{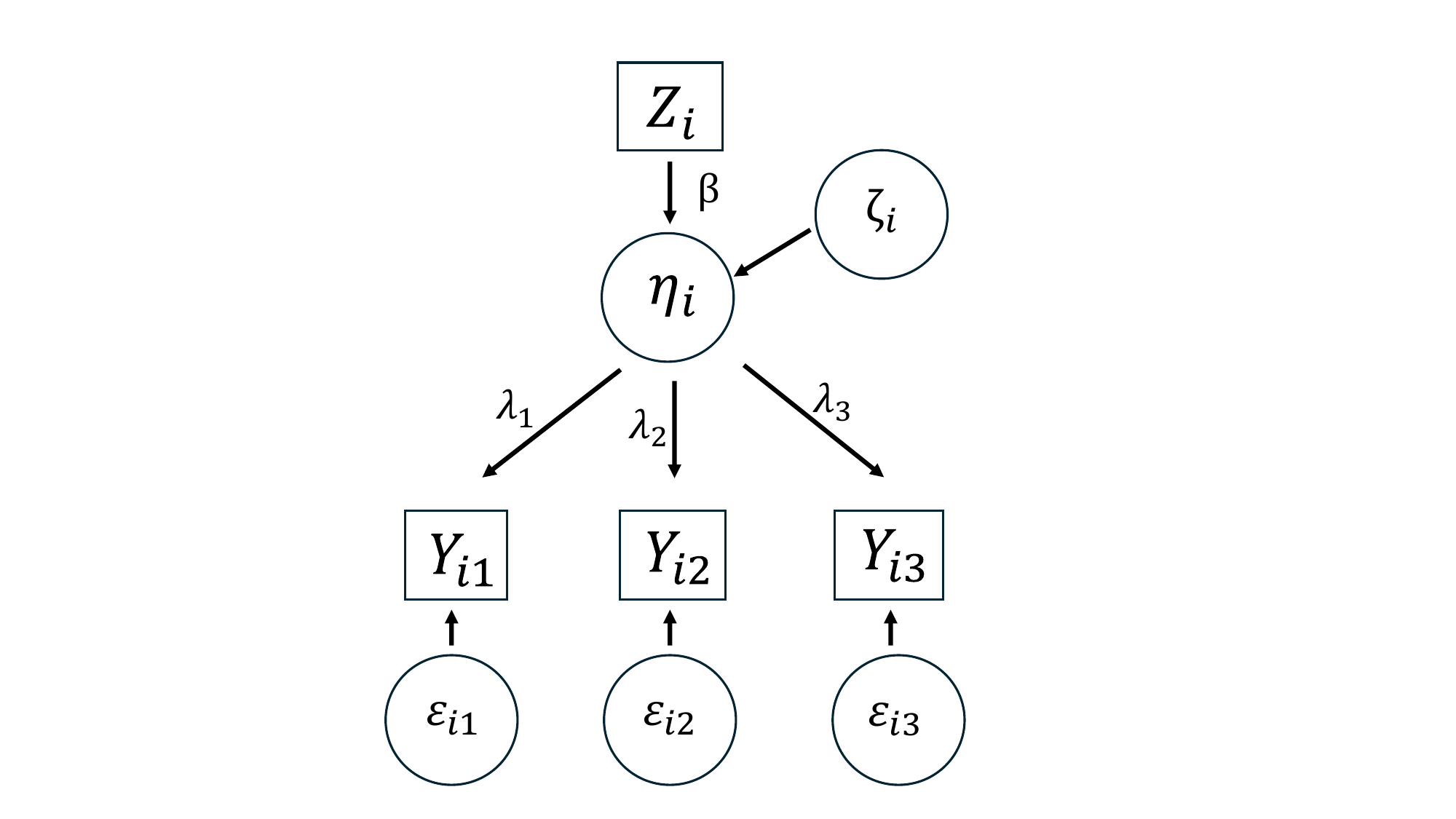}
    \caption{Graphical Depiction of an Experimental Design in which a Latent Outcome is Measured Linearly by Three Outcomes, Each Measured with Error.}
    \label{fig:dgp}
\end{figure}

Our identification strategy is rooted in a set of assumptions about the experimental design and the manner in which the latent variables are measured:

\begin{assumption}[Causal Framework for a Latent Outcome Variable]\label{ass:frame}

We assume $\forall j=1,2,...,J$, and $i=1,2,...,n$,

    A. Valid Measurement: $\lambda_j \neq 0$.
     
    B. $Z_i$ is randomly assigned: $Var(Z_i) \neq 0$ and $\{\eta^0_i,\eta^1_i\} \perp Z_i$

    C. $Z_i$ is excludable: $\{\epsilon^0_{ij},\epsilon^1_{ij} \}\perp Z_i$

    
\end{assumption}

The valid measurement assumption is straightforward: an outcome measure $Y_{ij}$ should contain information about the latent variable beyond pure noise. The second assumption requires that the treatment $Z_i$ has at least two distinct values and is randomly assigned to units, implying that it is independent of the potential latent outcomes. The assumption of excludability implies that each outcome's measurement errors are independent of the treatment. One potential scenario in which this assumption fails arises when the measured outcomes are affected by the treatment through channels other than the latent outcome of interest. We discuss this scenario and possible mitigation strategies in Appendix \ref{si:more}. 

This setup presupposes that the relationship between the latent variable $\eta_i$ and the observed measure $Y_{ij}$ is linear. We would offer four comments about linearity. First, our framework and results invoke linearity as a special case of a nonparametric identification strategy. (We discuss nonparametric identification and estimation in Section~\ref{sec:nonpara}.) Nonparametric identification typically requires an assumption that the measurements capture sufficient variability in the latent outcome $\eta_i$. This condition implies that discrete measurements with a small number of categories are generally insufficient. We argue that indexes 
constructed from multiple linear measures provide one of the simplest and most effective ways to capture the variability about the latent outcome.

Second, it is more natural for researchers to \emph{design} linear measurements than nonlinear ones. For example, test scores are typically constructed to increase with an individual’s latent ability, and respondents with more extreme latent attitudes are more likely to select extreme options on a seven-point Likert scale. It is rare for a designed measurement to exhibit a highly nonlinear relationship with the latent outcome unless the measured outcome is distributed in a truncated manner, with many observations falling into the highest or lowest categories. More granular measures tend to be more plausibly linear.

Third, beyond design considerations, linearity can also be evaluated empirically. If the linearity assumption holds, the measurements should satisfy $Y_{ij}=\lambda_j Y_{i1} + \epsilon_{ij} \; \forall j \neq 1$ (see lemma \ref{lem}). Thus, evidence of nonlinearity among measurements would lead us to reject the linearity assumption. As we will show in the \citet{kalla2020reducing} application below, the measures appear to exhibit the expected linear relationship. Were an outcome measure to exhibit signs of nonlinearity, the assumption of linearity can be relaxed for this problematic measure in ways that restore the desirable properties of the estimation approach, as illustrated in SI \ref{si:more1}.

Fourth, when outcomes are binary or ordinal, and we are not using nonparametric methods, we can still use a linear model to obtain a linear approximation; the problem, as we note below, is that assessing how well the model fits the data becomes more model-dependent. For this reason, our design recommendation is to gather outcomes in ways that satisfy the linearity assumption rather than shoehorn discrete outcome measures into a linear modeling framework. The formal results and corresponding strategies are discussed in the section \ref{sec:nonpara} and SI \ref{se:linearize}.

\subsection{Solving the Study-specific Noncomparability Problem}\label{sec:prob}

As emphasized in the introduction, because the latent variable $\eta_i$ has no inherent scale, arbitrarily standardizing measurements and reducing their dimensions -- as in methods such as PCA, ICW and IRT -- renders the latent outcome estimand study-specific. Even if two studies focus on the same latent outcome, any difference in the measurements will lead these methods to produce noncomparable results. As \citet[][pp. 28–29]{loehlin1998latent} points out, such standardization can lead researchers to mistakenly conclude that two estimated ALTEs differ when, in fact, they are the same. This poses a serious problem for the accumulation of knowledge.

Consider two very large studies in which the intervention exerts exactly the same average treatment effect on a latent outcome measured by the same indicators $Y_{ij}$. Standardization would produce two quite different estimates of the ALTE if the measurement error variances were much larger in the first study than in the second. (Indeed, the standardized estimates would differ in this example even though a simple regression of $Y_{i1}$ on $Z_i$ in each sample would, in expectation, yield identical estimates.) Standardization may even frustrate efforts to recover causal parameters from simulated data. In their assessment of several commonly-used methods of creating standardized indexes as experimental outcomes --- simple indexes, weighted indexes, and principal components factor scores --- \citet[][section 15.4]{blair2023research}  conclude that \emph{none} of them recovers a known parameter in repeated simulations.\footnote{In the SI \ref{si:decl}, we show that both of the estimators we propose below successfully recover the ALTE in their simulations (up to scale factor, depending on which outcome measure is used to set the metric).  Moreover, the estimated standard error from our approach closely matches the empirical standard error from their simulation.}  

Given these limitations, we propose an ``unstandardized'' approach. Because the latent variable $\eta_i$ lacks an intrinsic scale, we can, without loss of generality, set $\lambda_1=1$ so that we can interpret $\eta_i$ using the same metric as $Y_1$.  For example, if $\eta_i$ represents a latent distance, and $Y_{i1}$ is scaled in terms of kilometers while $Y_{i2}$ is scaled in terms of miles, setting $\lambda_1=1$ means that effects on $\eta_i$ are scaled in terms of kilometers \citep [pp.239-240]{bollen1989structural}. (The choice of which observed outcome to use for scaling purposes is arbitrary and has no effect on the statistical significance of the estimated ALTE.\footnote{Beware of the fact that SEM package in Stata uses its estimated standard errors to calculate $p$-values, but these standard errors are calculated based on numerical derivatives that are subject to scale-specific error. To obtain scale invariant $p$-values, conduct a likelihood ratio test comparing the log-likelihood of the fitted model to the log-likelihood of the restricted model, which constrains the average latent treatment effect to be zero.})

As a result, as long as one measurement is shared across studies, we can treat that measure as $Y_{i1}$. The latent treatment effect then targets the same latent outcome, allowing for direct comparison of results across studies. In short, maintaining interpretable scaling units rather than standardizing simplifies the identification problem and preserves the comparability of results across experimental replications. The remaining question is how to combine $Y_{i1}$ with other measurements that have different relationships with the latent outcome due to different $\lambda_j$ values. We address this issue in the next section.

\section{Identification of Measurement Parameters}\label{sec:identify}

We seek to identify the average causal effect of $Z_i$ on the latent variable $\eta_i$. However, $\eta_i$ is not directly observed, and even when setting $\lambda_1=1$ we can only partially measure $\eta_i$ given the unknown scaling parameters $\lambda_j$ and unobserved measurement errors $\epsilon_{ij}$. If we knew all the $\lambda_j$, we could rescale the observed measures, for example, by $\frac{Y_{ij}}{\lambda_j}=\eta_i+\frac{1}{\lambda_j}\epsilon_{ij}$, to approximate the latent variable $\eta_i$ using the same units as $Y_{i1}$.
The question is how to identify $\lambda_j$. As demonstrated in the following proposition, identification requires a third variable to serve as an instrumental variable. What is perhaps surprising is that this third variable can be either the treatment $Z_i$ or other measurements $Y_{ij}$ under the relatively mild conditions in Assumption \ref{ass:frame}.\footnote{When pre-treatment covariates are available, they can also serve as instrumental variables, as discussed in Section 6.1.}


\begin{proposition}[Identification of the measurement scaling parameters]\label{prop:iv}

Suppose Assumption \ref{ass:frame} holds. Then
in the above causal framework given $\lambda_1=1$, $\lambda_j$ is identified either by

(1) $\lambda_j=\frac{Cov(Z_i,Y_{ij})}{Cov(Z_i,Y_{i1})}$ \text{ if } $\mathbb{E}[\eta_i^1-\eta_i^0]\neq 0$ or by

(2) $\lambda_j=\frac{Cov(Y_{ik},Y_{ij})}{Cov(Y_{ik},Y_{i1})}\; \forall k \neq 1 \text{ and } k \neq j$, if $Cov(\epsilon_{ik}, \epsilon_{i1})=Cov(\epsilon_{ij}, \epsilon_{ik})=0$, $Var[\eta_i]\neq 0$, and $\{\eta_i^0,\eta_i^1\} \perp \{\epsilon^1_{ij},\epsilon^0_{ij}\}$.
    
\end{proposition}

\begin{proof}
The proof is in the SI \ref{si:allproof}.
\end{proof}

Proposition \ref{prop:iv} suggests that $\lambda_j$ can be identified via an instrumental variables approach. We can use either the treatment $Z_i$ or other outcome measures $Y_{ij}$ as instrumental variables when certain assumptions hold. Specifically, $Z_i$ must be a ``relevant'' instrumental variable whose average treatment effect on $\eta_i$ is nonzero (so that the denominator of the ratio converges to a nonzero value as the sample size increases). The excludability assumption of $Z_i$ as an instrument follows from Assumption \ref{ass:frame}B, because it is independent of the measurement errors. Similarly, $Y_{ik}$ is a valid instrumental variable if the measurement errors are uncorrelated with one another, and the limiting covariance between $Y_{ij}$ and $Y_{i1}$ is nonzero.  Even when the measurement errors are correlated, $Z_i$ remains a valid instrument so long as the treatment is unrelated to these errors of measurement.  

To illustrate the reasoning behind the instrumental variables approach, consider one strategy for identifying $\lambda_2$. From Lemma \ref{lem} in \ref{si:lem}, linearity between each measurement and the latent outcome implies linearity among the different measurements. In particular, 
$$
Y_{i2} = \lambda_2 Y_{i1} + (\epsilon_{i2}-\lambda_2\epsilon_{i1})
$$
Because $Y_{i1}$ is correlated with the error $\epsilon_{i1}$, $\lambda_2$ is not directly identified, and we cannot simply regress $Y_{2i}$ on $\eta_i$ because the latter is unobserved. Again, instrumental variables regression provides a consistent estimator. Given the excludability of the randomly assigned treatment variable $Z_i$ (Assumption \ref{ass:frame}C), we can use $Z_i$ as an instrumental variable for $Y_{i1}$. 
In the proof, we show that if the average treatment effect is nonzero  ($\mathbb{E}\tau_i \neq 0$), then  $Cov(Z_i,Y_{i1})\neq 0$. In that case, $Z_i$ is a ``relevant'' predictor of $Y_{i1}$. Moreover, given that $Z_i$ is randomly assigned and thus independent of measurement errors, $Z_i$ is a valid instrumental variable for $Y_{i1}$. Therefore, we can use $\frac{\widehat{Cov}(Z_i,Y_{i2})}{\widehat{Cov}(Z_i, Y_{i1})}$ to estimate $\lambda_2$. The same approach shows that $Y_{i3}$ can also serve as an instrumental variable for $Y_{i1}$ when we stipulate that the measurement errors are unrelated to the latent potential outcomes. This assumption implies, for example, that higher values of the latent potential outcomes are no more likely to coincide with higher values of the measurement errors than lower values of the latent potential outcomes.

What about cases such as Figure \ref{fig:dgp}, where the experiment includes \emph{both} a randomly assigned treatment and three measures of the latent variable?  In such cases, we have more than one plug-in estimator for $\lambda_2$ and $\lambda_3$, and therefore the unknown scaling parameters are said to be overidentified.  Overidentification is helpful in two ways.  First, we can combine the plug-in estimators to form a more efficient estimator, using method of moments or maximum likelihood. The former has the advantage of making weaker distributional assumptions about the observed or unobserved variables in model, while the latter imposes the assumption of multivariate normality.  In practice, the two estimators tend to produce similar estimates in large samples when the observed variables are distributed symmetrically \citep{olsson2000performance,browne1984asymptotically,yuan2005nonequivalence, boomsma2001robustness}.  

The second advantage of overidentification is that it allows us to test whether the data accord with the posited model. When two or more plug-in estimators render markedly different estimates, that is a sign that at least one modeling assumption is untenable. For example, when $Z_i$ is used as the instrumental variable for identifying $\lambda_2$, we need not invoke any assumptions about the uncorrelated measurement errors (i.e., the assumption that $\epsilon_{i1}$, $\epsilon_{i2}$, and $\epsilon_{i3}$ are uncorrelated with one another), whereas when we identify $\lambda_2$ using $Y_{i3}$ as the instrumental variable, we do invoke this assumption when asserting that the limiting covariances among the errors are zero.  If the two IV estimates differ markedly, that is a sign that uncorrelated measurement errors may be an untenable assumption.

\section{Estimation of the Average Latent Treatment Effect}\label{sec:est}

Once the $\lambda_j$ are identified, researchers have many options for estimating the average treatment effect on the latent variable. These options fall into two broad categories. The first category involves forming a weighted average of the observed outcome measures to approximate the latent variable, albeit with some residual error. This index-building approach has the advantage of allowing the analyst to use conventional methods, such as regression, to estimate the ALTE. An alternative approach is to estimate both the measurement parameters and the ALTE simultaneously as part of a system of linear equations. This approach is typically referred to as structural equation modeling or SEM. Although SEM is often associated with full-information estimators that assume multivariate normality \citep{joreskog1970general}, the same models can be rooted in more agnostic assumptions and estimated by method-of-moments \citep{kline2023principles}.  

This section discusses these two estimation approaches -- weighted scale indices (WSI) and maximum likelihood -- using simulated data to assess their ability to recover known treatment effects and estimate sampling variability. Fortunately, both approaches tend to perform well and generate similar results. The WSI has advantages in terms of transparency, as it allows the analyst to display regression results visually using individual-level data \citep{cook2009regression}; maximum likelihood has the advantage of providing a unified framework for estimating treatment effects and using overidentification to test modeling assumptions.

\subsection{Difference-in-means based on a Weighted Scaled Index (WSI)}

The simplest and most transparent estimator is the difference-in-means estimator using the weighted scaled index as the outcome. For each individual $i$, we first \textit{scale} each outcome measure using $\lambda_j$ to address the study-specific noncomparability problem. We then construct an outcome index measure by creating a \textit{weighted} average of the various outcome measures. Taking this weighted average to be the outcome, we use difference-in-means to estimate the causal effect. Formally, for each individual $i$, let $\tilde{Y}_{ij}=\frac{1}{\lambda_j}Y_{ij}$, and $\tilde{Y}_i=\sum_{j=1}^J \omega_j \tilde{Y}_{ij}$, where $\omega_j$ is the weight and $\sum_{j=1}^J \omega_j=1$. The most elementary weighting scheme is $\frac{1}{J}$. That is, the new outcome measure $\tilde{Y}_i$ is a simple average of all measurements for individual $i$. Although simple averages are widely used in practice \citep[e.g.,][]{ansolabehere2008strength}, they are not optimal from the standpoint of estimating the ALTE as precisely as possible.  A more efficient approach is to use inverse-variance weighting. Because we divide each measure by $\lambda_j$, the true variance of the measurement error is $\frac{\sigma^2(\epsilon^z_{\cdot j})}{\lambda^2_j}$; this leads to the optimal weight $\omega^*_j = \frac{\lambda^2_j/\sigma^2(\epsilon^z_{\cdot j})}{\sum_{j=1}^k \lambda^2_j/\sigma^2(\epsilon^z_{\cdot j})}$.
This optimal weighting scheme presented above assumes that measurement errors are independent; when errors are correlated, the optimal weighting algorithm becomes more complex.\footnote{In general, the weight is derived by minimizing the variance of the weighted average in treatment and control groups. For example, given each group, we minimize $\sum_{j=1}^J \omega_j \tilde{Y}_{ij}$ subject to the constraint that $1'w=1$. Define $\Sigma$ as the variance-covariance matrix of the measurement errors. By using the Lagrange multiplier method, we get the weight $w=(1'\Sigma^{-1}1)^{-1}(1'\Sigma)$, where $\Sigma$ is the variance and covariance matrix. If measurement errors are independent, the optimal  weight reduces to $\omega^*_j = \frac{\lambda^2_j/\sigma^2(\epsilon^z_{\cdot j})}{\sum_{j=1}^k \lambda^2_j/\sigma^2(\epsilon^z_{\cdot j})}$.} When the variance of the measurement error is assumed to be the same across the treatment and control groups, we can efficiently estimate a common $\sigma^2$ using the pooled data.



Using the WSI, a researcher applies conventional methods for estimating the ALTE. We begin by considering the difference-in-means estimator:
$$
\begin{aligned}
    \hat{\tau}&=\frac{1}{n_1}\sum_{i=1}^{n}Z_i\tilde{Y}_i - \frac{1}{n_0}\sum_{i=1}^{n}(1-Z_i)\tilde{Y}_i\\
    &=\frac{1}{n_1}\sum_{i=1}^{n}Z_i [\sum_{j=1}^k \omega_j\tilde{Y}_{ij}]- \frac{1}{n_0}\sum_{i=1}^{n}(1-Z_i)[\sum_{j=1}^k \omega_j \tilde{Y}_{ij}]
\end{aligned}
$$
Suppose we know the true $\lambda_j$. Then, it is straightforward to see that the oracle estimator is unbiased. When $\lambda_j$ is unknown, we can still use a consistent estimator of these scaling parameters to obtain a consistent difference-in-means estimator. 

\begin{proposition}\label{prop:unbias}
    Suppose assumption \ref{ass:frame} holds and all moments are finite. Consider a weighted difference-in-means estimator $\hat{\tau}=\frac{1}{n_1}\sum_{i=1}^{n}Z_i\tilde{Y}_i - \frac{1}{n_0}\sum_{i=1}^{n}(1-Z_i)\tilde{Y}_i$.
    
    (1) If $\lambda_j$ is known, then $\hat{\tau}$ is unbiased. 

    (2) If $\lambda_j$ is estimated consistently under proposition \ref{prop:iv}, then the weighted difference-in-means estimator is consistent.
    
\end{proposition}

\begin{proof}
The proof is in the SI \ref{si:allproof}.
\end{proof}

In practice, regression is widely used to estimate the ALTE, especially when covariates are used to improve precision. In SI \ref{si:stack}, we show that the above difference-in-means estimator is equivalent to a ``stacked'' regression. 

Let us now consider the variance of this estimator.  If the $\lambda_j$ scaling factors are known based on the measurement properties observed in prior studies (and $\omega$ is pre-specified), then we can apply the conventional Neyman variance estimator, as discussed in SI \ref{si:var}. If the $\lambda_j$ and/or $\omega_j$ are estimated from the data, we have to account for the estimation uncertainty. The standard approach, as with the \textit{inverse probability weighting estimator}, uses  generalized method of moments (GMM), as discussed in SI \ref{si:gmm}.

\subsection{Structural Equation Modeling}

An alternative approach is to model all of the parameters in Figure \ref{fig:dgp} as a system of linear equations, one equation for each of the endogenous variables in the causal graph. Figure \ref{fig:dgp}, for example, implies one equation for the latent outcome variable ($\eta_i$) and three equations for each of the observed outcome measures.  After imposing the scaling metric $\lambda_1=1$, we are left with a total of 8 free parameters: the variance of the randomized treatment ($\phi$), the average treatment effect of $Z_i$ on $\eta_i$, the variance of the unobserved causes of the latent outcome variable ($\psi$), two scaling parameters ($\lambda_2$ and $\lambda_3$), and three variances of the measurement errors associated with $Y_{i1}$, $Y_{i2}$, $Y_{i3}$. As we show in SI \ref{si:struct}, these 8 parameters may be used to predict the 10 observed variances and covariances among the four measured variables ($Z_i$,$Y_{i1}$,$Y_{i2}$,$Y_{i3}$). Since some of these parameters are overidentified, the estimator we use may affect our empirical estimates because different estimators use different objective functions when gauging the fit between the estimated parameters and the observed variance-covariance matrix.  As noted above, the most commonly used estimator is maximum likelihood under the assumption of multivariate normality, but other estimators that make weaker distributional assumptions are available and tend to perform well when the number of observations is large relative to the number of parameters \citep{yuan2005nonequivalence}.  


\subsection{Comparison to Current Methods}

Our WSI method consistently estimates latent treatment effects that are comparable across studies. Commonly used approaches for handling multiple measurements include principal components analysis and inverse covariance weighting. Table in SI \ref{si:summarytab} lists 14 recent papers published in the \emph{American Political Science Review} that employ these two methods. However, unlike SEM (which is rarely used to analyze experiments), these methods were not originally designed to estimate average latent treatment effects; inverse covariance weighting was designed to facilitate hypothesis testing, and principal components analysis was designed to reduce the dimensionality of a collection of correlated measures. We compare these methods with WSI and demonstrate that our approach achieves additional statistical precision.

Figure \ref{si:powermain} compares the power of five alternative methods for constructing outcome measures.\footnote{\citet{stoetzer2022causal} propose a hierarchical item response theory (hIRT) model to estimate the latent treatment effect. Their model has two components. First, the latent variable $\eta_i$ is assumed to be a linear function of the treatment: $\eta_i = \gamma_0 + \gamma_1 Z_i + \epsilon_i$, where $\epsilon_i$ is normally distributed with a constant variance $\sigma^2$. In this equation, $\gamma_1$ is the average treatment effect of interest. Next, the observed outcomes $Y_{ij}$ are generated by a characteristic function of the item: $\mathbb{P}[Y_{ij}=h | \eta_i]=P_{ijh}(\eta_i;\alpha_{ijh},\beta_{ij})$, where $\alpha_{ijh}$ and $\beta_{ij}$ are the item difficulty and item discrimination parameters for item $j$ and answering $h$. As \citet{stoetzer2022causal} note (e.g. page 28), this model makes a set of parametric assumptions, including constant treatment effects across subjects and a particular item characteristic function. Our approach relaxes some of these assumptions. We allow $Z_i$ to have heterogeneous treatment effects. We also remain agnostic about the non-linear transformation function. Because the IRT model is not suitable in the presence of heterogeneous treatment effects (HTE), we do not report power calculations. More discussion may be found in SI \ref{si:comp2}.} WSI and SEM estimators prove to have the greatest power (they are almost indistinguishable in the figure). Both methods are clearly superior to an equally weighted scaled index, because optimal weights increase the signal-noise ratio. The WSI and SEM also tend to be more powerful than PCA and ICW. Ironically, those dominated methods are the most frequently used approaches in recently published work. Detailed discussion of the simulations may be found in Appendix \ref{si:comp1}.

\begin{figure}[!h]
    \centering
    \includegraphics[width=0.8\linewidth]{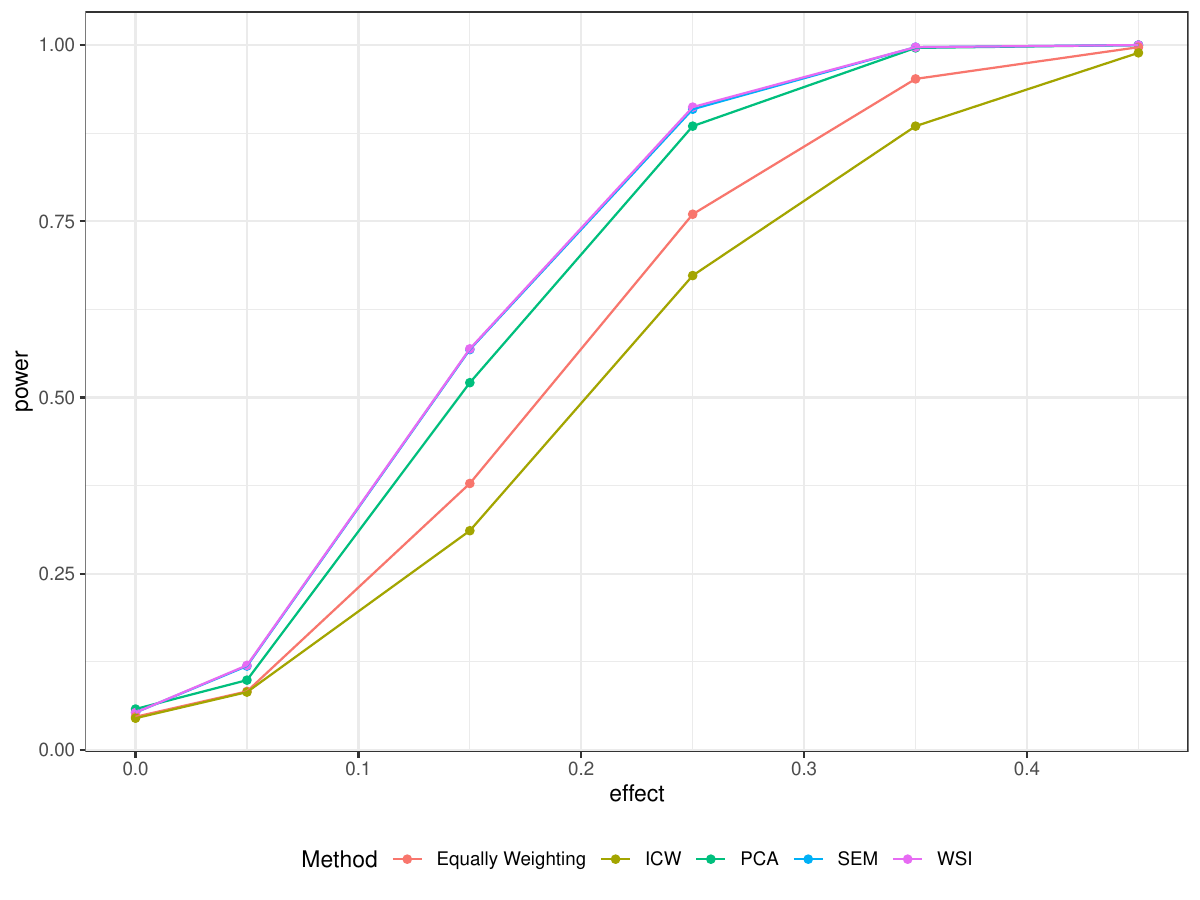}
    \caption{Power analysis: WSI, Equal weighting, SEM, ICW, and PCA}
    \label{si:powermain}
\end{figure}

\textbf{t-ratio and noise in the outcome variable.} Why does suboptimal measurement of experimental outcomes affect statistical power?  Recall the regression function in the framework: $\eta_i = \mathbb{E}[\eta_i^0] + \tau Z_i + \xi_i$. For simplicity, suppose $\lambda_j = 1 ; \forall j$. Recall that the weighted (scaled) outcome is $\tilde{Y}_i = \sum_{j=1}^J \omega_j Y_{ij}$. Therefore, the observed regression $\tilde{Y}_{i}=\mathbb{E}[\eta_i^0] + \tau Z_i + (\tilde{\xi}_i+\epsilon'_i)$, where $\tilde{\xi}_i=\sum_{j=1}^J \omega_j \xi_{ij}$ and $\tilde{\epsilon}_i=\sum_{j=1}^J \omega_j \epsilon_{ij}$.
The $t$-ratio for the treatment variable is $t=\frac{\hat{\tau}}{\hat{se}(\hat{\tau})}$, where the denominator $\hat{se}(\hat{\tau})$ depends on the estimated variance of the error term: $\hat{se}(\hat{\tau})=\sqrt{\frac{\hat{\sigma}^2_{\tilde{\xi}}+\hat{\sigma}^2_{\tilde{\epsilon}}}{\sum_{i}(Z_i-\overline{Z}_i)}}$. From this expression, we see that when the noise in the outcome measure is large, the $t$-ratio becomes smaller.

\section{Design Trade-off: More Outcome Measures or More Subjects?}\label{sec:var}

One implication of the preceding section is that, as \citet[][p.449]{broockman2017design} argue, researchers could improve the precision of their experiments by investing in additional outcome measures. Under what conditions is such an investment warranted?

The super-population variance of the (oracle) estimator is 
\begin{equation}
    Var(\hat{\tau}) =\frac{Var[\tilde{Y}^1_i]}{n_1}+\frac{Var[\tilde{Y}^0_i]}{n_0}
\end{equation}

This formula shows that researchers can decrease this variance by adding more measurements, which decreases the variance in the numerator, or by adding more observations, which increases the denominator. We examine both approaches and explore the optimal allocation under a budget constraint.

\subsection{Advantages of Adding More Outcome Measurements}

The variances of the difference-in-means or regression estimators depend mainly on the variance of the weighted averages $\tilde{Y}^1_i$ and $\tilde{Y}^0_i$. Assuming $Cov(\epsilon_{ij},\epsilon_{ik})=0$ and ignoring the estimation variance of $\hat{\lambda}$ on the grounds that it is known based on existing studies, under optimal inverse-variance weighting $\omega^*_j$, the variance of $\tilde{Y}^1_i$ is $\sigma^2(\eta^1_i) + \frac{1}{\sum_{j=1}^J \lambda^2_j/\sigma^2(\epsilon_{\cdot j})}$, where the latter part is due to (optimally weighted) measurement error.


%
The variance formula shows that the inclusion of a highly reliable outcome measure (i.e., a measure with a high ratio of trait variance to error variance) may substantially increase the precision with which the ALTE is estimated. However, the marginal gains from each additional measure depend on the reliability of the measures already contained in the additive index. Fortunately, with optimal weighting, in expectation more measurements never increase the variance of the estimated ALTE because unreliable measures are down-weighted accordingly.\footnote{It is important to note that this desirable feature of adding outcome measures may not hold without optimal weighting. See more discussion in \ref{si:variance}. } 





\subsection{Gathering More Observations}

In the super-population framework, the variance of our WSI estimator is shown in Equation (1),  
the numerators are the variances of weighted average potential outcomes in the super-population. Adding more observations affects the denominator. Under a balanced design, when one adds $j$ more observations in each group, the variance will be $[\frac{2Var[\tilde{Y}^1_i]}{n+2j}+\frac{2Var[\tilde{Y}^0_i]}{n+2j}] / [\frac{2Var[\tilde{Y}^1_i]}{n}+\frac{2Var[\tilde{Y}^0_i]}{n}] =\frac{n}{n+2j}$ of the variance when the sample size is $n$; in other words, the percentage reduction in variance is $\frac{2j}{n+2j}$. Therefore, researchers can add $\frac{n}{4}$ subjects to the treatment and control groups to decrease the variance by $\frac{1}{3}$.

\subsection{Trade-off under a budget constraint} 

We can apply this framework to a more general budget allocation problem. Suppose that researchers have a budget $B > 0$ and the cost of adding one more measure is $c_m$ and the cost of adding one more observation is $c_o$. Consider the optimal allocation of budget for the sample size $n$ and the number of items $J$. This is a typical constrained optimization problem 
where objective is to minimize the sampling variance given the budget constraint $c_mJ+c_0 n\leq B$. Optimization may be used to find the best budget allocation between additional outcome measures and additional respondents.

Now, let us consider a more concrete scenario. Suppose researchers have an extra  $\$5000$ budget. They may decide to recruit more respondents and/or add more measures. What is the optimal decision? To be consistent with the previous simulation, we assume the current sample size is $n=500$, and there is only one measure. The marginal cost of the measure is $c_m=\$1000$ and the marginal cost of each respondent is $c_o=\$10$. Let us consider two scenarios: the outcomes are measured with high reliability ($0.75$) or low reliability ($0.4$).\footnote{To be specific, we set $Var(\eta^1)=Var(\eta^0)=1$. For low reliability, the variance of the measurement error is $1.701$ so that $\frac{Var(\eta)}{Var(Y)}\approx 0.4$. For high reliability, the variance of the measurement error is $0.379$ so that $\frac{Var(\eta)}{Var(Y)}\approx 0.75$. The objective function is to maximize the variance gain, or equivalently, $\min_{k,j} \frac{-2k[Var(\eta^1)+Var(\eta^0)]}{500(500+k)}+\frac{-4(k+500j+jk)}{500(1+j)\Sigma(500+k)}$, where $j$ is the number of additional measures and even number $k$ represents the additional respondents, and $\Sigma:=\frac{\lambda^2}{\sigma^2(\epsilon)}$. The constraint is $1000j + 10k \leq 5000$.}

For high reliability case, the optimal solution is to add one more measure and recruit 400 more respondents. However, for the low reliability case, the optimal solution is to add two more measures and recruit 300 more respondents. In other words, although the extra measures are low quality, they are an especially good investment. This somewhat surprising result is formally explained in SI\ref{si:measure}; but it can be intuitively understood from the simulation in Figure \ref{fig:red}. When measurement error is large (low reliability), the variance reduction is relatively greater.

\begin{figure}[!h]
    \centering
    \includegraphics[width=0.8\linewidth]{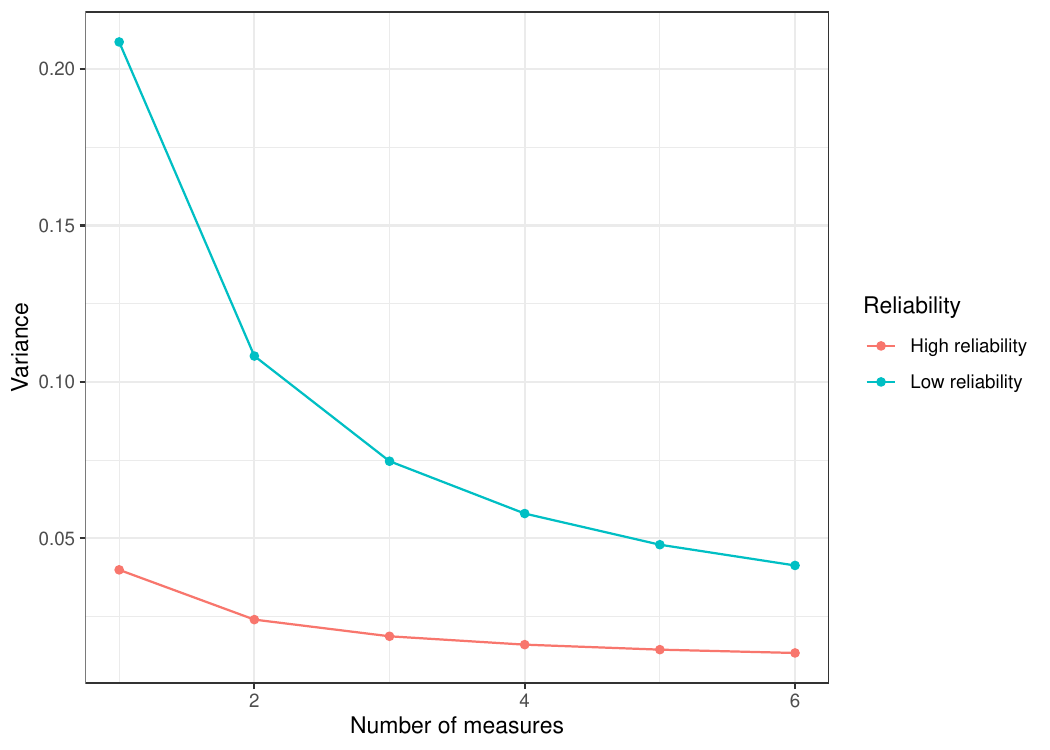}
    \caption{\textbf{Simulation Illustrating Variance Reduction Due to Collection of Additional Outcome Measures.} The horizontal line represents the number of measures and the vertical line is the estimated variance of the optimal weighting estimator. The variance reduction is larger if the reliability is lower. A detailed explanation may be found in SI \ref{si:measure}.}
    \label{fig:red}
\end{figure}



Two other considerations also arise when weighing the advantages of collecting additional outcome measures.  The first is that additional observable manifestations of the latent outcome make for more overidentifying restrictions on the key scaling parameters.  Each additional measure gives us more ways to assess the scaling properties of all of the outcome measures at hand.  This fact, in turn, allows us to assess the goodness-of-fit of a posited multi-equation model.  The second consideration is that the more outcome measures one gathers, the more likely one is to discover the inadequacies of one or more of the measures. One symptom of a measurement problem is poor fit between the predicted and actual variance-covariance matrix; another is fluctuation of the estimated ALTE when certain outcome measures are included or excluded.  Ideally, additional outcome measures improve precision and robustness; however, the inclusion of flawed measures may undercut these benefits.  Therefore, when gauging the trade-off between allocating resources to subjects or outcome measurements, one should bear in mind that the value of additional outcome measures hinges on their validity and reliability.

\section{Strategies for Improving Precision and Assessing Robustness}
\subsection{Covariates} 

The analysis of experiments with latent outcomes may benefit in three ways from the inclusion of covariates, defined as variables measured prior to random assignment.  First, covariates may be used to identify measurement parameters, such as the $\lambda_j$ in Figure \ref{fig:dgp}.  Indeed, when it comes to identifying scaling parameters, covariates play the same identifying role as randomly assigned interventions, and the assumptions required are similar: the covariate must be statistically independent of the measurement errors but correlated with the structural disturbance term ($\zeta_i$).  In other words, the covariate must predict $Y_{i1}$ to some extent so that the instrumental variables estimator is defined.  Second, covariates that are strongly prognostic of outcomes will improve the precision with which the average treatment effects are estimated.  This property of covariate adjustment is in keeping with conventional experimental analysis, where outcomes are modeled directly without reference to latent variables.  Third, because covariates contribute to the identification of measurement parameters, covariates allow the researcher to use overidentification to assess the goodness-of-fit of the posited multi-equation model. In sum, although covariates are not strictly necessary for the identification or estimation of measurement parameters, they may nevertheless play a useful role, especially when they are predictive of outcomes.

\subsection{Nonparametric Identification and Linearity}\label{sec:nonpara}

A key assumption of our approach is the linear relationship between the latent outcome $\eta_i$ and its observed measurements $Y_{ij}$. As emphasized in previous sections, this is fundamentally a design-based assumption: researchers are responsible for designing appropriate measures of the latent construct. The simplest and most natural way to measure a latent outcome is by positing a linear relationship -- for instance, test scores should increase with latent ability. It is difficult to justify a measurement that has a highly nonlinear relationship with the latent outcome, especially when a simpler linear measure is available.

Our framework, in fact, allows for more flexible nonparametric extensions. We establish the formal nonparametric results—both for technical rigor and completeness—in a separate paper, and sketch the main ideas here. Instead of assuming a linear relationship, each measurement $Y_{ij}$ can be viewed as a random draw from some distribution indexed by the latent variable $\eta_i$, denoted $P_j(\cdot|\eta_i)$. The randomness naturally implies the presence of measurement errors. Our linear specification, $Y_{ij}= \lambda_j \eta_i + \epsilon_{ij}$ is a special case in which the expectation of $Y_{ij}$ is linear in $\eta_i$. 

To address the study-specific noncomparability problem (as in the linear case), we need to rescale all measurements relative to $Y_{i1}$. Specifically, we require a measurement bridge function $h_j$ for each measurement $j\neq 1$ such that $\mathbb{E}[h_j(Y_{ij})|\eta_i]=\mathbb{E}[Y_{i1}|\eta_i]$. This bridge function maps each measurement onto the reference measurement $Y_{i1}$. In the linear case, the bridge function is unique and takes the form $h_j(Y_{ij}) = \frac{1}{\lambda_j} Y_{ij}$, such that $\mathbb{E} [\frac{1}{\lambda_j} Y_{ij} \mid \eta_i] = \mathbb{E}[Y_{i1} \mid \eta_i]$. 

One sufficient condition of the existence of such a measurement bridge function, $h_j$, is completeness assumption.\footnote{The completeness assumption and results may be found in the double negative control literature \citep{miao2018confounding,tchetgen2020introduction}. Formally, the completeness assumption states that $\mathbf{E}[g(\eta_i)|Y_{ij}]=0$ almost surely if and only if $g(\eta_i)=0$ almost surely.} Intuitively, completeness requires that the measurement $Y_{ij}$ captures the full information (variability) of the latent outcome $\eta_i$. For example, if $\eta_i$ is discrete with three categories, then $Y_{ij}$ must also take on at least three distinct values. In the typical case in which $\eta_i$ is continuous, we do not recommend using measurements that lose information (such as discretized or categorical outcomes). Instead, we encourage researchers to design measures that plausibly bear a linear relationship to the underlying latent variable of interest, which is the simplest way to satisfy the completeness assumption. 

As in the proximal causal inference literature \citep{miao2018identifying,miao2018confounding}, identifying $h_j$ requires auxiliary information from either $Z$ or another measurement $Y_{ik}$ (for $k \neq 1 \text{ and } j$). For example, by taking the expectation of both sides of the bridge equation with respect to $p(\eta_i|Z_i)$, we obtain $\mathbb{E}[h_j(Y_{ij})|Z_i]=\mathbb{E}[Y_{i1}|Z_i]$. 
In the linear case, with the bridge function  $h_j(Y_{ij})=\frac{1}{\lambda_j} Y_{ij}$, we obtain 
$Y_{ij} = \lambda_j Y_{i1} + (\epsilon_{ij}-\lambda_2\epsilon_{i1})$ and use $Z_i$ as an instrumental variable to identify $\lambda_j$, as shown in section \ref{sec:identify}. In more general settings, $h_j$ can be estimated using nonparametric methods.

\begin{example}[Binary measurements]
    Suppose researchers measure a latent outcome through one continuous measurement $Y_{i1}$ and two binary measurements $Y_{ij}, j=2,3$. Binary measures are an instructive case because they cannot be linear functions of the latent variable. 
    
    Recall, $h_j$ is the solution satisfying
    $\mathbb{E}[h_j(Y_{ij})|Z_i]=\mathbb{E}[Y_{i1}|Z_i]$. Because $Y_{ij},j=2,3$ only takes two value, for $Z_i=1$, we obtain 
$$
\begin{aligned}
    \mathbb{E}[Y_{i1}|Z_i=1]&= \mathbb{E}[h_j(Y_{ij})|Z_i] \\
   &= h_j(Y_{ij}) \mathbb{P}[Y_{ij}=1|Z_i=1] + h_j(Y_{ij}) \mathbb{P}[Y_{ij}=0|Z_i=1] \\
   &= h_j^1 \mathbb{P}[Y_{ij}=1|Z_i=1] + h_j^0 \mathbb{P}[Y_{ij}=0|Z_i=1].
\end{aligned}
$$

Here $h_j^1$ and $h_j^0$ are two unknown parameters. Similarly, for $Z_i=0$, we obtain 
$$
\mathbb{E}[Y_{i1}|Z_i=0]=h_j^1 \mathbb{P}[Y_{ij}=1|Z_i=0]+ h_j^0 \mathbb{P}[Y_{ij}=0|Z_i=0].
$$

\noindent We have two equations and two unknowns, and $h_j^1$ and $h_j^0$ are identified under NPIV assumptions. 

To estimate the ALTE, we first obtain transformed outcomes $\tilde{Y}_{i2}=h_2(Y_{i2})$ and $\tilde{Y}_{i3}=h_3(Y_{i3})$. In this binary case, we can write as $\tilde{Y}_{i2}=Y_{i2}h_2^1+(1-Y_{i2})h_2^0$ and $\tilde{Y}_{i3}=Y_{i3}h_3^1+(1-Y_{i3})h_3^0$. The final combined outcome variable is $\tilde{Y}_i=\omega_1 Y_{i1}+\sum_{j=2}^3\omega_j \tilde{Y}_{ij}$. With this new outcome variable, we can estimate ALTE as usual, for example $\mathbb{E}[\tilde{Y}_i|Z_i=1]-\mathbb{E}[\tilde{Y}_i|Z_i=0]$ \footnote{In practice, we need to use additional instruments; otherwise, there is no efficiency gain in this example, because the sample moments are exactly the same for the three outcomes given treatment status.}. 

To illustrate this identification result, we simulate $\eta_i^0 = 0$ and $\eta_i^1 \sim N(1,1)$. The outcome $Y_{i1}$ is then constructed as $Y_{i1} = \eta_i + e_{i1}$. For the binary measurements $Y_{i2}$ and $Y_{i3}$, we first generate latent continuous variables $\xi_{i2} = 0.5\eta_i + e_{2i}$ and $\xi_{i3} = 1.5\eta_i + e_i$, and then create binary measurements based on whether the latent variable $\xi$ is above or below its mean: $Y_{i2} = 1[\xi_{i2} < \overline{\xi_2}]$ and $Y_{i3} = 1[\xi_{i3} < \overline{\xi_3}]$. Figure \ref{fig:nonp} shows that nonparametric WSI correctly estimates the true ALTE of 1. In contrast, if we ignore the nonlinearity and instead use linear WSI or SEM, the estimates are clearly biased.  The implication is that a nonparametric approach allows the researcher more latitude in estimating the ALTE when some of the outcomes are binary indicators.

\begin{figure}
    \centering
    \includegraphics[width=0.5\linewidth]{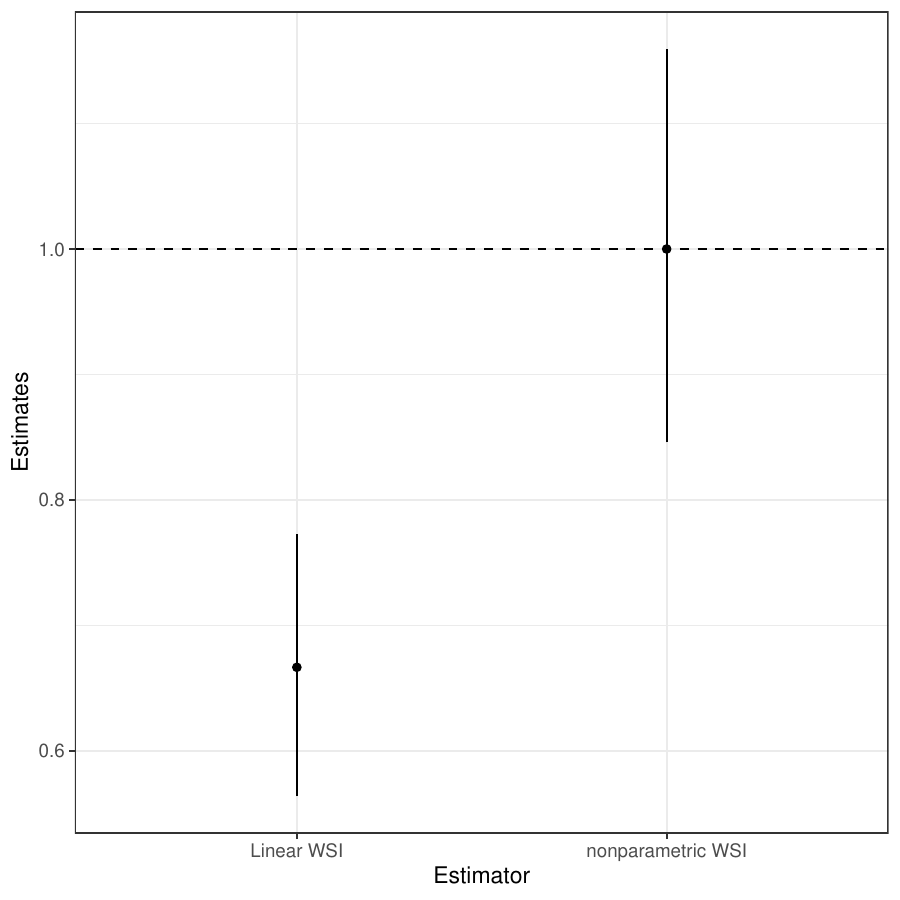}
    \caption{Binary measurements and Nonparametric Estimation}
    \label{fig:nonp}
\end{figure}

\end{example}


Although nonparametric methods provide less model-dependent results, they inevitably come with larger estimation error. We therefore suggest that researchers begin by testing for linearity. Lemma \ref{lem} shows that linearity between the measured outcomes and the latent outcome implies linearity among the different measured outcomes. Therefore, we can test whether $Y_{ij}$ and $Y_{ik}$ follow a linear functional form. The literature offers multiple specification tests for this purpose. For example, the Rainbow test checks whether linearity holds in the central portion of the data. \citep{utts1982rainbow}; the Regression Equation Specification Error Test  examines whether the coefficients on added polynomial terms are jointly zero \citep{ramsey1969tests}.
More generally, we can compare the linear specification to any nonparametric alternative using methods proposed by \citet{hardle1993comparing}, \citet{zheng1996consistent}, \citet{hsiao2007consistent}, among others.

As noted above, linearity is untenable when latent outcomes are measured by survey questions with binary response options. One work-around is to use more granular response options, such as presenting an agree-disagree question with response options ranging from ``strongly agree'' to ``strongly disagree.''\footnote{However, this approach may still yield a lopsided response distribution that would not plausibly be modeled as a linear function of the latent variable.} In the extreme case where only a few discrete measurements are available, a more convincing way to satisfy the linearity condition is to create an additive index composed of several discrete measures, in much the same way that standardized tests measure math ability by counting the number of correct answers to many specific multiple-choice questions.\footnote{Modern variants of scholastic aptitude tests adaptively calibrate the difficulty of the questions based on the answers to prior questions, and the same cost-saving principle can be applied to survey measures of political knowledge and other topics \citep{montgomery2022adaptive}.}
This approach, which dates back to the early days of psychometric assessment, may require pilot testing (perhaps in a nonexperimental context or during pre-treatment baseline measurement) in order to devise a series of discrete measures that, when summed, produce granular distributions with few observations in each tail. For technical discussion, please see SI \ref{se:linearize}.

\begin{figure}[!h]
    \centering
    \includegraphics[width=.9\linewidth]{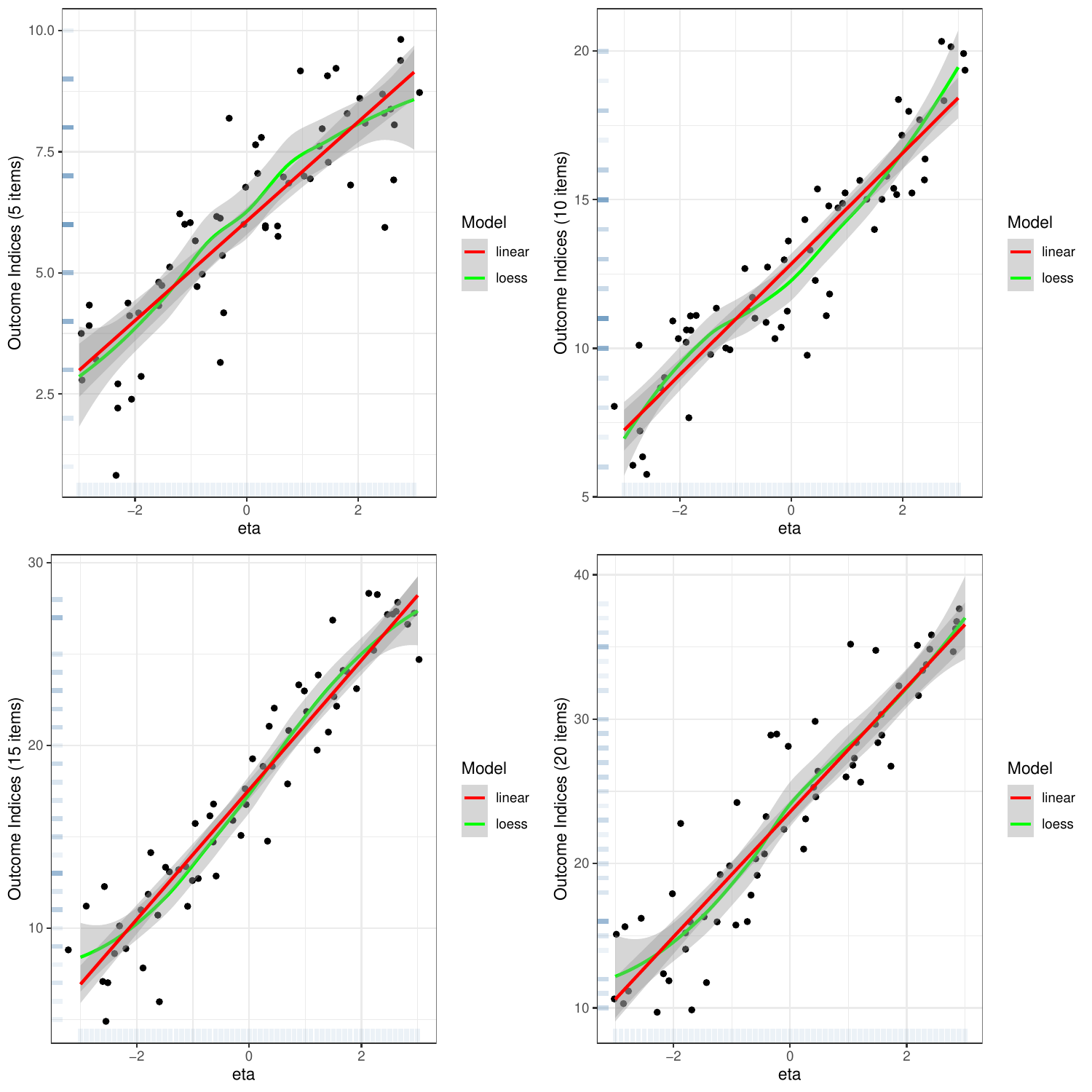}
    \caption{\textbf{Assessing linearity between $\eta$ and an additive index created by adding binary items together}. We create an additive index $v_1$ by summing up 5, 10, 15, and 20 binary responses from IRT models. Data points have been jittered slightly for clarity. The rug plots on both axes denote the distribution of data. The vertical axis shows the additive index created by adding all binary variables in the simulation. The horizontal axis is the true latent variable ($\eta$) used in the IRT model.}
    \label{fig:lin1}
\end{figure}

To illustrate how additive indices help satisfy the assumed linear mapping between latent and observed outcomes, we performed a series of simulations. In each simulation, binary responses were generated using an IRT model, and two outcome indices were then created by summing half of the binary items together.\footnote{We use simTrt function in the psych R package to generate the data. To be specific, the discrimination parameter $a \sim Uniform (.5, 1.5)$, item difficulties $b \sim N(0, 2)$, and the guessing asymptote is $c=0.2$.} The number of items used to create each index increases from one simulation to the next, as indicated by the horizontal axis labels of Figure \ref{fig:lin1}. The sequence of graph panes shows that as we sum more binary items, the LOESS fitted line -- which allows a nonlinear relationship if the data suggest it -- between the observed measures and the true $\eta_i$ used in the simulation becomes increasingly linear. Figure \ref{fig:lin1} illustrates an extreme case where all component items are binary. In SI \ref{se:linearize}, we also show the simulation results for ordinal items. Intuitively, if the component items have more granular response options, fewer items will be needed to form an additive index that meets the linearity condition.\footnote{In practice, $\eta$ is not observed. Instead, we could show the increasingly linear properties between two observed measures. Figure \ref{fig:lin2}, which does so, conveys the same point: as the number of binary items used to create an index increases, the more linear the apparent relationship.}  On the other hand, building an index from survey items that are subject to the same measurement errors due to similar question wording, order, and response format slows the rate at which linearity is achieved.  As a thought experiment, consider the limiting case in which two measures elicit exactly the same responses from each subject; in this case, an index based on both measures is no more granular than each measure considered separately. The design implication is clear: if possible, researchers should measure the latent variable in different ways so as to reduce the correlation between errors of measurement \citep{andrews1984construct}.

\subsection{Testing Measurement Equivalence across Experimental Groups}

Our basic framework assumes that responses from subjects in the treatment and control groups have the same measurement structure. In other words, the translation from $\eta_i$ to the $Y_{ij}$ is the same regardless of the subject's experimental condition. This assumption may not hold in practice, for example, if there is something about the treatment that changes the way that subjects interpret outcome questions.\footnote{The intuition extends to non-survey outcomes as well. If outcomes in the control group are measured in miles but outcomes in the treatment group are measured in kilometers, any apparent mean difference may be misleading.}  Different scaling parameters jeopardize the exclusion restriction, which holds that the treatment itself is the only systematic reason that expected outcomes in the treatment and control groups may differ.  Fortunately, the measurement equivalence assumption can be assessed empirically. In order to make the test as general as possible, we relax the implicit assumption that the observed outcome is centered on the latent variable. A more flexible model is that $Y_{ij} = \lambda^Z_j g(\eta_i) + \alpha_z + \epsilon_{ij}$, where $\alpha_z$ denotes the so-called ``structural mean'' \citep{bagozzi1989use} for experimental group $z$. Unlike the model presented in section \ref{sec:frame}, this one expresses the treatment effect as a shift in latent means, while at the same time allowing for the possibility that the measurement scaling parameters $\lambda^Z_j$ 
may differ for the treatment and control groups. The more flexible model may be tested as a nested alternative to the model that assumes the same measurement parameters apply to both experimental groups, as illustrated in the empirical example below.

\subsection{Comparison to Seemingly Unrelated Regression (SUR)}

Many experiments are designed with a clear intention to gather multiple measures of a given latent outcome, in which case the modeling framework described here is an appropriate starting point.  However, other experiments gather multiple outcomes that are not necessarily viewed as redundant measures of a given latent factor.  For example, field experiments on conditional cash transfers examine outcomes such as whether children in the household are enrolled in school, vaccinated, and free from signs of malnutrition.  One could view such outcomes as manifestations of an underlying factor (child well-being), but one could instead consider them to be three distinct outcomes.  How might we evaluate the adequacy of the two modeling approaches?

\begin{figure}[!h]
    \centering
    \includegraphics[width=0.8\linewidth]{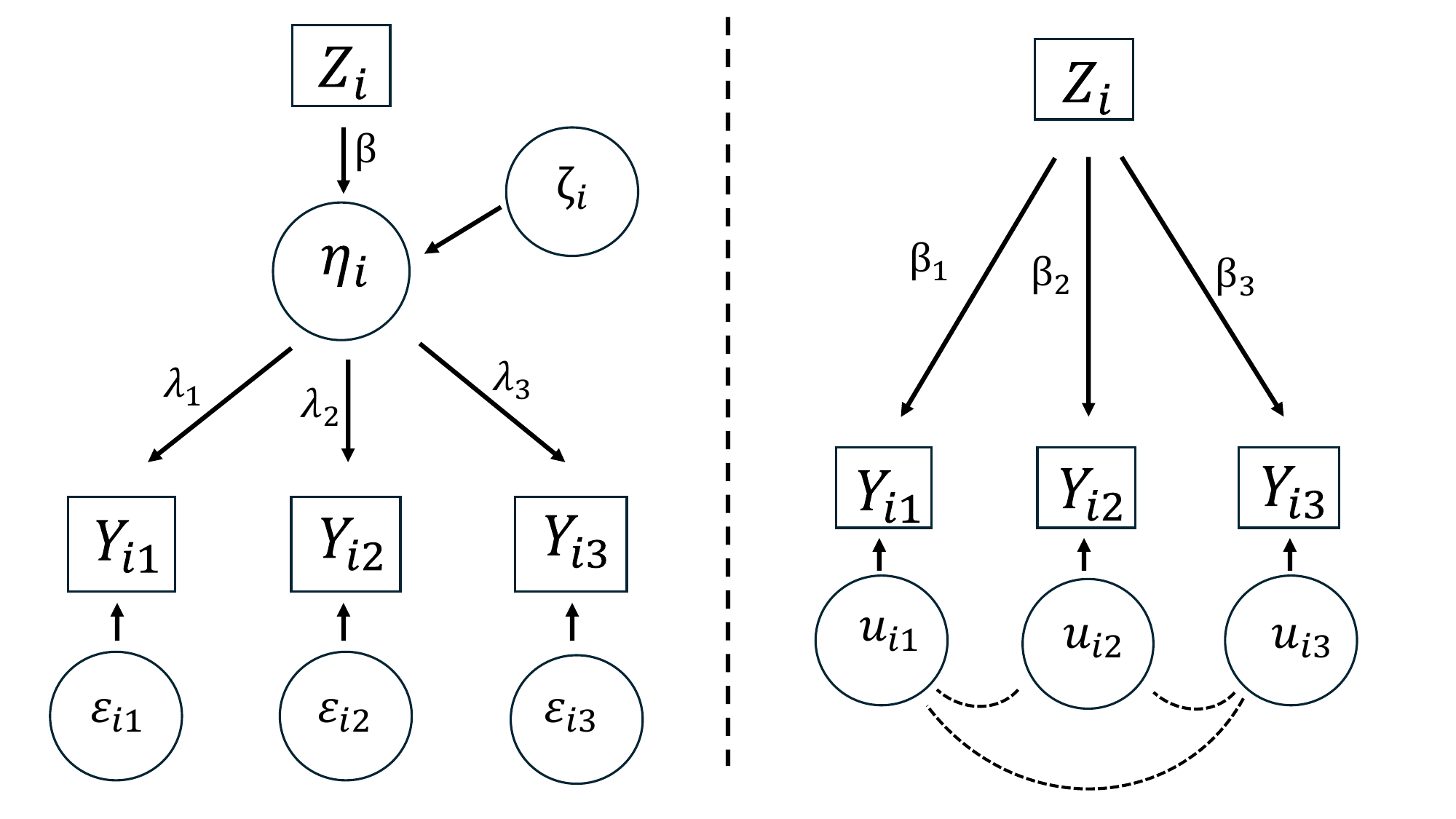}
    \caption{Graphical Depiction of a Linear Measurement Process with Additive Measurement Errors (left) and a Seemingly Unrelated Regression Model (right)}
    \label{fig:dgpsur}
\end{figure}

Fortunately, the two models may be expressed as nested alternatives. Figure \ref{fig:dgpsur} shows two causal graphs side-by-side.  The left pane depicts the now-familiar model in which three outcome measures are linear manifestations of a latent outcome. The right pane depicts the seemingly unrelated regression model \citep{zellner1962efficient} in which the same three outcome measures are modeled as distinct outcomes.\footnote{The inverse covariance weighting method proposed by \citet{anderson2008multiple} is closely connected to the SUR model because it does not posit a latent variable measured with error.  This method  down-weights outcomes that are correlated with one another in order to generate a unified outcome measure with minimum variance. Simulations in Figure \ref{si:powermain} show that the inverse covariance weighting method is less powerful than SEM when outcomes measures truly tap a latent outcome with error.  
}  Whereas the latent variable model has 8 free parameters, the SUR model has 10 because it does not constrain the causal path from $Z_i$ to the $Y_{ik}$ to flow through $\eta_i$. Therefore, two degrees of freedom may be used to assess the extent to which the more parsimonious latent variable model adequately fits the observed variance-covariance matrix.  Rejection of the null hypothesis calls into question one or more assumptions of the latent variable model in favor of the more agnostic SUR model. Conversely, failure to reject the null hypothesis suggests that the causal pathways implied by the latent variable model are approximately the same as the three distinct average treatment effects depicted in the SUR model (i.e., $\beta \lambda_k \approx \beta_k$). On theoretical grounds, the latent variable model may not be applicable to every experiment that measures multiple outcomes.  But when the latent variable model is plausible for a given application, the fact that its parameterization has testable implications means that researchers need not rely solely on theoretical intuition when choosing between the latent variable model and the SUR model.

\section{Application}

Like \citet{stoetzer2022causal}, we draw our empirical example from the \citet{kalla2020reducing} field experiment designed to assess the effects of two different door-to-door canvassing interventions on subjects' views concerning immigrants. Kalla and Broockman's Study 1 shows that canvassing conversations that feature a ``non-judgmental exchange of narratives'' produce persistent changes in subjects' attitudes about immigrants, whereas otherwise similar conversations that omit this persuasive strategy seem to produce weaker effects.  One attractive feature of this study is the fact that it devoted considerable attention and resources to outcome measurement.  As summarized in Table \ref{tab:summary}, five questions gauged respondents' attitudes towards undocumented immigrants.  A separate battery of three questions measured respondents' views about policy proposals concerning expediting paths to citizenship and reducing threats of deportation.

The two treatment conditions (full treatment $Z_1$ and abbreviated treatment $Z_2$) in conjunction with two primary outcome scales (attitudes $Y_1$ and policy views $Y_2$) and a pre-treatment covariate give us five observed variables. If we model these data by assuming that the attitude index and the policy views index both measure a common underlying factor, the two parameters of central interest are the average treatment effect of full treatment on the latent outcome and the average treatment effect of abbreviated treatment on the latent outcome.  The latent variable model is depicted in the left pane of figure \ref{fig:dgpsur1}.

Let us review and evaluate the assumptions under which both causal parameters in the latent variable model are identified: the treatments are randomly assigned and therefore independent of latent potential outcomes; the treatments affect measured outcomes only insofar as they influence the latent outcome (the exclusion restriction, which implies that the measurement errors for each observed outcome ($\epsilon_{ij}$) are independent of the interventions); the stable unit treatment value assumption (SUTVA), which implies no spillovers between subjects; at least one of the two treatments has a nonzero ATE on the latent outcome, or the covariate is predictive of outcomes (but not the measurement errors); and observed outcomes are linear functions of the latent outcome. Next, let us evaluate the plausibility of these assumptions.

\begin{figure}[!h]
    \centering
\includegraphics[width=0.8\linewidth]{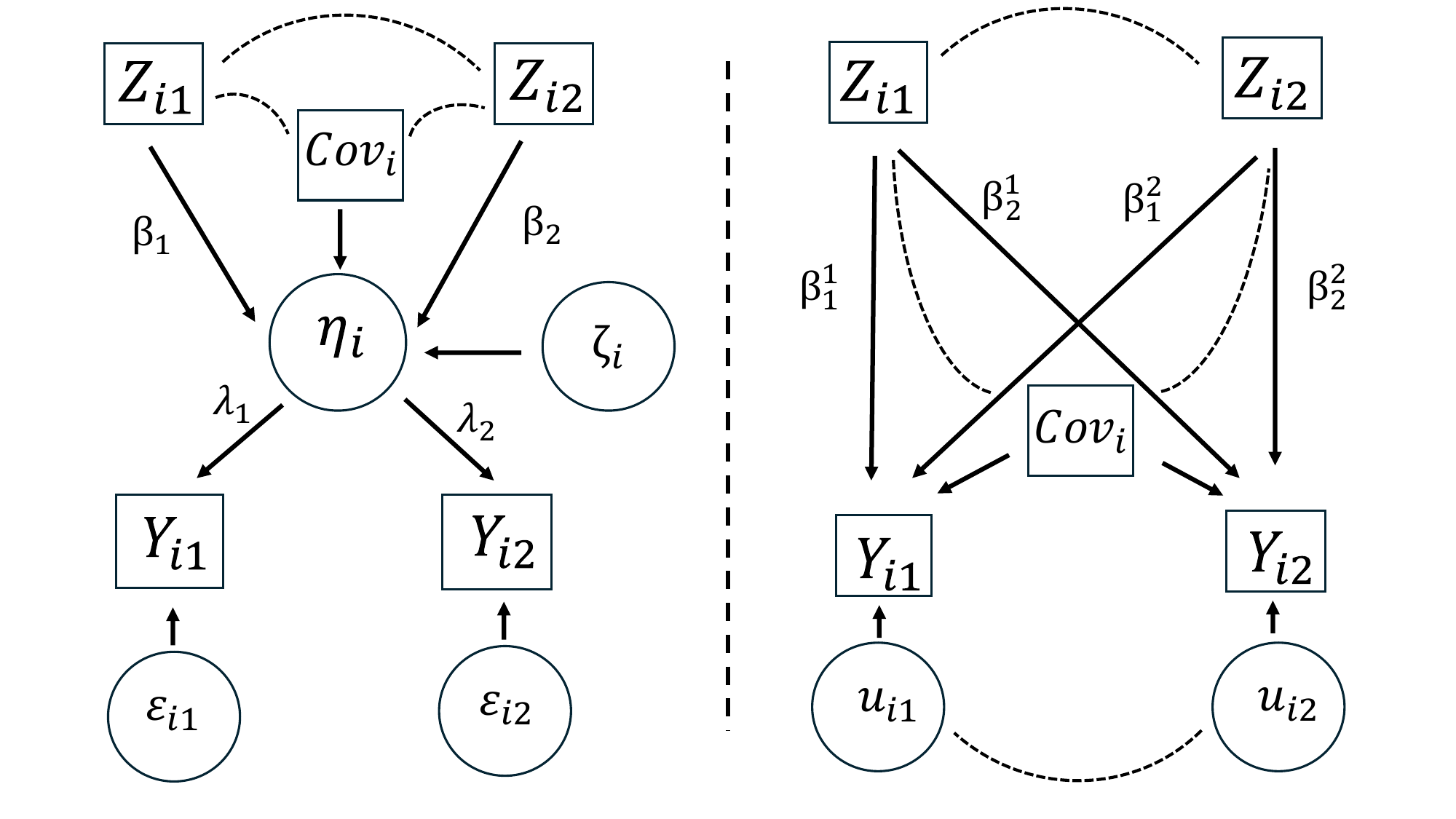}
    \caption{Graphical Depiction of \citet{kalla2020reducing} with Additive Measurement Errors (left) and a Seemingly Unrelated Regression Model (right).}
    \label{fig:dgpsur1}
\end{figure}

In this application, random assignment is satisfied by the experimental design and remains tenable in the wake of attrition, which seems to be unrelated to treatment assignment.  The exclusion restriction implies that the interventions do not affect the outcome measures except insofar as they affect the underlying factor \citep{stoetzer2022causal}.  This assumption could be jeopardized if subjects drew the connection between the treatment they received and the survey, but this study used an unobtrusive measurement strategy that did not alert participants to the connection between the series of surveys they completed and their incidental visit from a canvasser.  SUTVA violations due to spillover effects seem unlikely to have a material effect on outcomes here given the separation among subjects from different households. The treatments and covariates are clearly predictive of outcomes, so the scaling parameters are identified.

As for the linearity assumption, one of the study's strengths is its use of multi-item outcome measures. These measures gauge outcomes with sufficient granularity that we can visually inspect the scatterplot between $Y_1$ and $Y_2$. As shown in \ref{fig:lin}, fitting a flexible LOESS curve through this plot reveals a near-linear pattern that coincides with the fitted least-squares regression line.  A linear mapping from the latent outcome to each of the observed outcome measures seems plausible in this application.\footnote{Tables A.2 and A.4 may be used to calculate the proportion of observed variance in the outcomes that is attributable to the latent variable $\eta_i$. For attitudes, this proportion is 0.74, and for policy views it is 0.83.}


With these identifying assumptions, consistent estimation is straightforward. 
If we declare the scale of the latent outcome variable to be in the same units as the attitude index (i.e., $\lambda_1 = 1$), the ATE of the full treatment on the latent factor can be estimated consistently using the approach described in Section \ref{sec:est}.\footnote{Because this application features two randomly assigned interventions ($Z_{i1}$ and $Z_{i2}$) that are correlated, consistent estimates are obtained using multiple regression rather than bivariate regression.}  
The ALTE of the abbreviated treatment is consistently estimated in an analogous fashion. The fact that we have more information than we need to identify certain parameters such as $\lambda_2$ 
means that can (1) compare alternative estimates of each parameter and (2) leverage the overidentifying information to obtain more precise estimates. Comparing alternative estimates is one way to assess model fit; when two or more estimates of the same parameter diverge, the implication is that one or more underlying assumptions are incompatible with the data.  

Estimates of the two interventions' average treatment effects on the latent outcome are presented in Table \ref{tab:reg}. The first two columns present estimates obtained using maximum likelihood, and columns (3) and (4) report WSI estimates.\footnote{When creating the weighted average, we estimate the scaling parameter using method-of-moments. The optimal weight $\omega^*_j = \frac{\hat{\lambda}^2_j/\hat{\sigma}^2(\epsilon_{\cdot j})}{\sum_{j=1}^k \hat{\lambda}^2_j/\hat{\sigma}^2(\epsilon_{\cdot j})}$, where the variance of measurement error is calculated from MLE.}  Columns (5) and (6) report SEM estimates, this time using a GLS estimator rather than MLE and using bootstrapped standard errors. Each estimator is presented using two specifications, one without covariate adjustment and one adjusting for respondents' baseline attitudes about immigrants (See SI \ref{si:app}).  Comparing columns (1) and (3) shows that MLE and WSI generate almost identical estimates. 
Comparing columns (2) and (4) again shows that the two estimators render similar estimates. Covariate adjustment greatly improves precision and clearly indicates that the full treatment was effective in changing subjects' opinions. The MLE covariate-adjusted estimates also suggest that the full treatment was more effective than the abbreviated treatment, with a differential effect of 0.341 (SE=0.141).  The MLE estimates and standard errors are almost identical to those obtained using GLS and bootstrapped standard errors.


\begin{table}[!htbp]
\begin{threeparttable}
\caption{Estimation Results: MLE, Weighted Index, and GLS}\label{tab:sem}\label{tab:reg}
\centering
\begin{tabular}{@{\extracolsep{5pt}}lcccccccc} 
\\[-1.8ex]\hline 
\hline 
\\[-1.8ex] & \multicolumn{6}{c}{} \\ 
 & {MLE}   & {MLE}     & {WSI}    & {WSI}   & {GLS}  & {GLS}   \\
 \\[-1.8ex] & (1) & (2) &  (3) & (4) & (5) & (6) \\ 
 \hline \\[-1.8ex] 
treat (full) & {0.384*} & {0.431***}  & {0.384*} & {0.430***}  & {0.384*}&{0.431***} \\
 & (0.213) &(0.099) &  (0.196) & (0.099) & (0.213) & (0.096)\\
treat (brief)  & {0.224}  & {0.090}        & {0.225}  & {0.090}  & {0.224}  &  {0.090}   \\
 & (0.206) & (0.101) & (0.200) & (0.101) &(0.215) & (0.102)\\
baseline covariate    &    & {0.662***} &                    & {0.662***} & & {0.662***}\\
& & (0.011)  & & (0.009) & &  (0.012)\\
$\lambda_2$ & 1.653 & 1.549 & 1.652 & 1.549 & 1.653 &1.549 \\
\hline \\[-1.8ex] 
$\chi^2$ $p-value$  & 0.923    & 0.978     &   &   & 0.923 &0.978 \\   
Adjusted $R^2$     &    &   & 0.001  &  0.762  & &  \\
Degree of freedom  &1 & 2 & & & & \\
\hline 
\hline \\ 
\end{tabular}
 \begin{tablenotes}     
\item {\it Note}: Sample size $N=1,578$. $\lambda_2$ in columns 1, 2, 5, and 6 are estimated by SEM; $\lambda_2$ in columns 3 and 4 are estimated by 2SLS using the two treatments and, where applicable, covariates as instrumental variables. For GLS, the standard errors are computed from 10,000 bootstraps. The MLE models are estimated by R package lavaan, and the GLS models are estimated by R package systemfit. *** $p<0.01$, ** $p<0.05$, and * $p<0.1$.  
 \end{tablenotes}
\end{threeparttable}
\end{table}

Next, we compare the results from Table \ref{tab:reg} to estimates obtained using the seemingly unrelated regression model, which makes no assumptions about why the two outcomes may be correlated.  
Table \ref{tab:sur} presents two different SUR specifications, with and without covariate adjustment. The first column, without covariate adjustment, reports two separate OLS regressions in which attitudes ($Y_1$) and policy views ($Y_2$) are regressed on the two treatment indicators.  The second column repeats these two regressions, this time including the baseline covariate. Although the estimated treatment effects look different in Table \ref{tab:reg} and Table \ref{tab:sur}, some simple manipulations show that the two tables tell very similar stories.  The estimated effects on $Y_{i1}$ are scaled identically in the two tables, and the two estimates are quite close.
The estimated effects on $Y_{i2}$, however, differ because the latent variable model applies that scaling factor $\widehat{\lambda}_2=1.653$ when estimating the ALTE in column (1) of Table \ref{tab:reg}.  If we divide the SUR estimate by this estimate of $\lambda_2$, we come very close to the ALTE estimate reported in Table \ref{tab:reg}.  The same holds for all of the coefficients reported in the two tables, implying that the modeling constraints of the latent variable model fit are quite compatible with the agnostic SUR model, which by construction fits the observed variance-covariance matrix perfectly.  A more formal test of the fit of the latent variable model against the fit of the SUR model comes to the same conclusion: the $p$-value of the likelihood-ratio test is $0.980$, so we have no basis to reject the adequacy of the (more parsimonious) latent variable model.  

\begin{table}[!htbp]
\caption{Estimation Results: Seemingly Unrelated Regression}\label{tab:sur}
\begin{center}
\begin{tabular}{l c c }
\\[-1.8ex]\hline 
\hline 
 & \multicolumn{2}{c}{} \\ 
 & Model 1  & Model 2  \\
\\[-1.8ex]  \hline

equation 1: treat (full)  & $0.383$       & $0.418^{***}$ \\
                  & $(0.214)$     & $(0.120)$     \\
equation 1: treat (brief)  & $0.231$       & $0.088$       \\
                  & $(0.218)$     & $(0.122)$     \\
equation 2: treat (full)  & $0.635$       & $0.689^{***}$ \\
                  & $(0.335)$     & $(0.192)$     \\
equation 2: treat (brief)   & $0.364$       & $0.143$       \\
                  & $(0.342)$     & $(0.196)$     \\
equation 1: baseline covariate &               & $0.662^{***}$ \\
                  &               & $(0.011)$     \\
equation 2: baseline covariate &               & $1.025^{***}$ \\
                  &               & $(0.018)$     \\
\hline
eq1: R$^2$        & $0.002$       & $0.688$       \\
eq2: R$^2$        & $0.002$       & $0.673$       \\
eq1: Adj. R$^2$   & $0.001$       & $0.687$       \\
eq2: Adj. R$^2$   & $0.001$       & $0.673$       \\     
\hline
\hline
\end{tabular}
\begin{tablenotes}     
\item {\it Note}: In Model 1, two equations represent two regression models: eq1 is attitudes $\sim$ treat (full) + treat (brief) and eq2 is policy views $\sim$ treat (full) + treat (brief). In Model 2, we add the baseline covariate.  *** $p<0.01$, ** $p<0.05$, and * $p<0.1$.      
\end{tablenotes}
\label{table:coefficients}
\end{center}
\end{table}

Our final specification check is to assess whether the measurement properties of the outcome measures operate symmetrically for all three experimental groups.  The null hypothesis is that the measures operate in the same way across the three assigned groups; the alternative model is that the measurement parameters -- the scaling factors and the measurement error variances -- differ across groups.  A total of 6 degrees of freedom differentiate the two models.  The $p$-value of the model comparison is 0.171, suggesting that there may be some concern about differential measurement across groups, but the gap in terms of goodness-of-fit is not decisive.\footnote{In their analysis of differential item functioning when applying a hierarchical IRT model to these data, \citet{stoetzer2022causal} found similarly equivocal results.}


\section{Conclusion}

This paper has shown how experiments with redundant outcome measures may be analyzed using latent outcome models while retaining the agnosticism of a design-based framework. We identify a key problem of current methods for multiple measurements -- their estimands are often study-specific, hindering the accumulation of scientific knowledge. In contrast, our WSI estimator preserves comparability across experimental replications while at the same time offering other advantages, such as improved statistical precision.

Although our nonparametric identification approach is agnostic about whether the observed measures are linearly related to the latent outcome, in practice linearity is an attractive modeling assumption.
Whether linearity is a plausible assumption for a given application depends on what is being measured and how; crucially, linearity may be addressed during the design stage of an experiment and assessed empirically. For example, outcome measures may be constructed from batteries of survey questions, and the resulting indices may be sufficiently granular to warrant a linear modeling framework. In such cases, identification and estimation of causal effects become straightforward. In other words, an investment in outcome measurement pays dividends when it comes time to analyze the experimental results with less reliance on restrictive parametric assumptions, such as those invoked by an IRT model \citep{stoetzer2022causal}.

A further advantage of our framework is that it allows the researcher to make use of statistical tools that are routinely used in experimental analysis, ordinary least-squares regression and instrumental variables regression.  Instrumental variables regression may be used to estimate the scaling parameters that link the latent outcome to the observed outcomes, which in turn may be used to estimate the error variances associated with each observed outcome measure. With estimates of these scaling parameters and error variances in hand, a researcher may easily create an optimally weighted average of the observed measures of the latent outcome that has an interpretable scale. This WSI may be regressed on the treatments and covariates in the usual way, and the relationship between interventions and the latent outcome are easily visualized using scatterplots. Alternatively, structural equation models may be used to estimate the scaling parameters, error variances, and causal parameters in a single statistical procedure. Both approaches produce consistent estimates of the average treatment effect on the latent outcome. 

Another important message of this paper is the value of overidentification via additional outcome measures, treatments, or covariates. These layers of redundancy both facilitate robust estimation of scaling parameters and create opportunities for goodness-of-fit tests that can help assess the empirical adequacy of the posited latent variable model. Failure to pass such tests implies that future studies need to invest in new or improved outcome measures.  

Latent variable models are optional but potentially helpful. When an outcome cannot be observed directly, researchers should strive to design experiments that measure this latent outcome in multiple ways. For decades, statisticians have pointed out that one way to increase the power of an experiment is to measure the outcome more reliably. The algebra in section \ref{sec:var} reiterates this point, showing the gains in precision that occur with each additional measure of a latent outcome. When maximizing the precision with which the average treatment effect on the latent outcome is estimated, it sometimes makes more sense to invest in additional outcome measures rather than additional subjects.  

This paper has focused on the challenge of measuring a single latent outcome, but the framework may be expanded to designs in which two or more latent outcomes are posited (See the SI \ref{si:more} for an example). The development of valid and reliable measures of multiple latent constructs has a long history in psychometric research \citep{anderson1988structural,campbell1959convergent}. The details of this type of investigation go beyond the scope of this paper but typically involve an iterative process of proposing measures that seem to tap into a theoretically defined construct (face validity), properly distinguish subjects that are known to differ on the latent dimension of interest (construct validity), and show the expected patterns of high and low correlations with measures of other constructs (convergent and discriminant validity) \citep{chan2014standards}. The systematic study of measurement involves a combination of theoretical and empirical steps that may be conducted outside the confines of an experiment but potentially provides valuable insights for experimental applications.

\putbib 
\end{bibunit}


\newpage


\clearpage

\appendix
\addcontentsline{toc}{section}{References} 
\part{Supplementary Information} 
\parttoc 

\setcounter{figure}{0}
\setcounter{table}{0}
\setcounter{proposition}{0} 
\renewcommand\thefigure{A.\arabic{figure}}
\renewcommand\thetable{A.\arabic{table}}
\renewcommand\theproposition{A.\arabic{proposition}}
\renewcommand\thelemma{A.\arabic{lemma}}

\onehalfspacing
\setcounter{page}{1}
\begin{bibunit}

\newpage
\section{An Overview of Current Practice for Constructing Outcome Indices in Political Science}\label{si:summarytab}

\begin{table}[!ht]
\centering
\caption{Illustrative Articles with Multiple Outcome Measurements (2020-2025) Published in the \emph{American Political Science Review} }
\begin{tabular}{p{5 cm} p{1.2cm} p{3cm} p{5.8cm}}
\hline
\textbf{Citation} & \textbf{Year}  & \textbf{Method} & \textbf{Latent variable(s)} \\
\hline
\citet{cruz2020social} & 2020  & ICW  & Public goods provision \\
\citet{blair2022liberal} & 2022  &ICW and equal weights &   Policy provisions \\
\citet{arias2024eye} & 2024  & ICW &   Climate policy related latent outcomes  \\
\citet{haas2024my} & 2024  & ICW  &  Overall positive evaluation of a politician. \\
\citet{bowles2023sustaining} & 2024  & ICW  &  Multiple primary outcomes relating to COVID-19 and politics \\
\citet{amar2025countering} & 2025  & ICW &  Seven distinct primary outcomes \\
\citet{tavits2024fathers} & 2024  & PCA  &  Gender-equality attitudes \\
\citet{abramson2022historical} & 2022& PCA + IRT  &  Social trust; political trust \\
\citet{fouirnaies2022electoral} & 2022  & PCA  &  Legislator effort \\
\citet{peterson2022influence} & 2022  & PCA &  Tax opinion, perceived polarization, cyber security\\
\citet{cheeseman2022curse} & 2021  & PCA  &  Pessimism about corruption prevalence \\
\citet{SobolevEtAl2020} & 2020  & PCA  &  Protest size \\
\citet{thal2020desire} & 2020  & PCA  &  Desire for social status \\
\citet{barker2022intellectualism} & 2022 & PCA  &  Several latent psychological concepts \\




\hline
\end{tabular}
\begin{tablenotes}  
\item \textit{Note:} Abbreviations refer to methods for constructing outcome indices, using multiple outcome measures as inputs. ICW = inverse covariance weighting of standardized items. Equal = all scale items are standardized and receive equal weighting. PCA = principal components analysis, where the index is constructed based on first dimension factor loadings. IRT = Item response theory.
\end{tablenotes}  
\end{table}

\newpage

\section{Revisiting a Simulation to Recover the ALTE}\label{si:decl}

In their discussion of multiple outcomes, \citet[][section 15.4]{blair2023research}, consider several index-creation methods for outcome measurement.\footnote{This example is drawn from code in section 15.4 of the book, which focuses on estimating the conditional mean of a latent variable.  We build on the authors' data-generating process to study the related question of estimating an average treatment effect of an experimental intervention using what is effectively a difference-in-means estimator.  Note that we could have reparameterized the SEM model as a structured means model in order to estimate the intercepts as well.}
Their simulation model posits an experimental intervention on a latent outcome, the same problem that we consider in this paper.
However, none of their index-creation methods  -- additive standardized indexes or predicted values based on factor scores -- allows the authors to successfully recover the target ALTE from their simulation. In this section, we apply our two proposed estimators to their setup: the weighted scaled index and the SEM (maximum likelihood) estimator. 

In their simulation, the latent variable is generated by the equation $\eta=1 + X + 2 * rnorm(N)$. The true ALTE of the experimental intervention is therefore 1. The three observed outcome measures are generated as follows: $Y_1=3+0.1*\eta+rnorm(N,sd=5)$,$Y_2=2+1*\eta+rnorm(N,sd=2)$, and $Y_1=1+0.5*\eta+rnorm(N,sd=1)$.  Notice that the three scaling factors ($\lambda_j$) are $\{0.1, 1, 0.5\}$, so if we set $\lambda_2=1$, we will get outcomes in the original metric. If we instead set $\lambda_1=1$, the model fit will be identical, but now the scale in which the ALTE is measured will change by a factor of 10. The outcome measure we use to set the scale is arbitrary (e.g. millimeters vs. centimeters), but we must remember to interpret the ALTE in terms of whatever scale we select.

Figure \ref{fig:decl} shows the empirical distribution of our two estimators across 1000 simulations, each with a sample of 500 subjects. For ease of interpretation, we set $\lambda_2=1$ so that the target ALTE parameter remains 1.0.  Both of these estimators produce estimates that are centered around the true average treatment effect. Further examination of the empirical sampling distributions shows that both approaches work well, with SEM doing a slightly better job of estimating the true standard error.

\begin{figure}[!h]
    \centering    \includegraphics[width=0.6\linewidth]{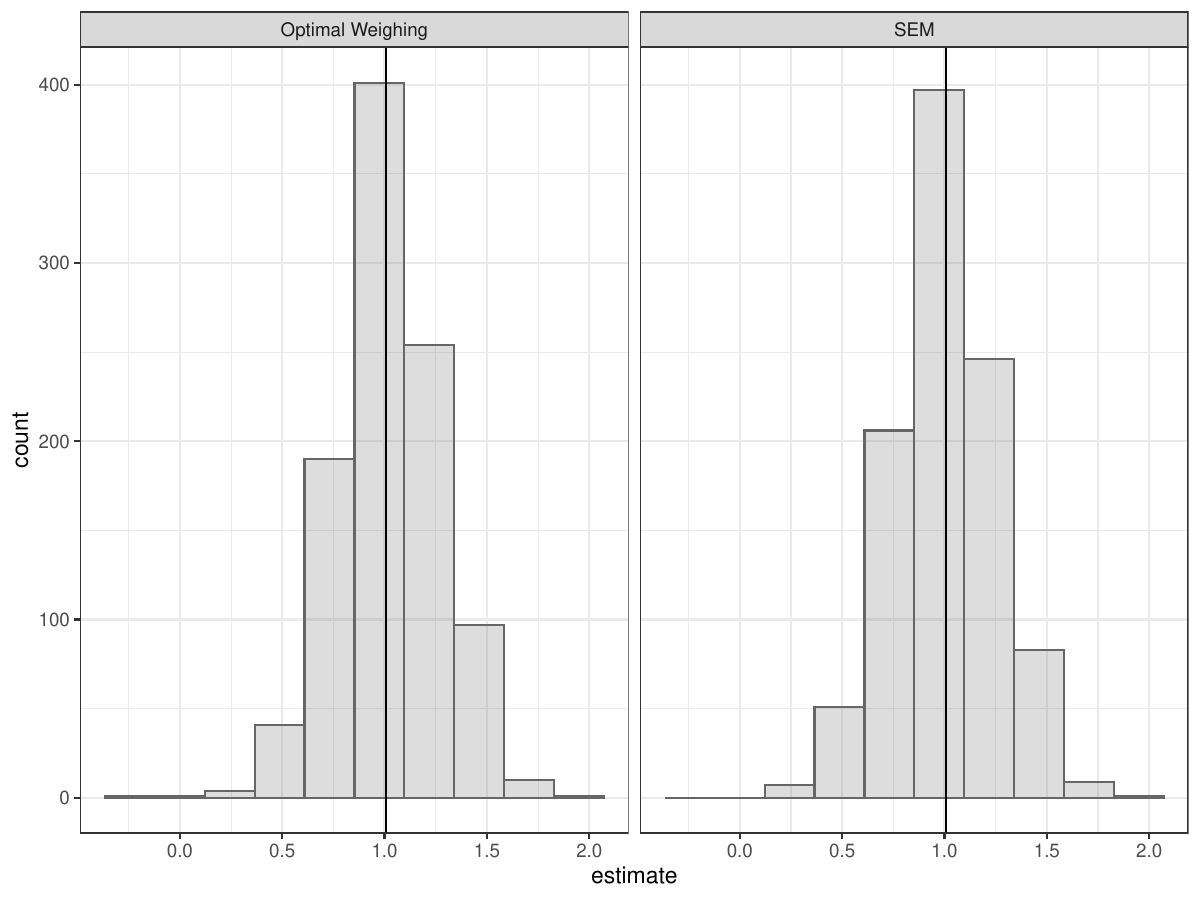}
    \caption{\textbf{Empirical distribution of optimal-weighted scaled index and SEM estimators.} The vertical line is the target ALTE of the intervention on the latent outcome, which is scaled to equal to one in the simulation.}
    \label{fig:decl}
\end{figure}

\section{Proof of Main Propositions}\label{si:allproof}

\subsection{Linearity Lemma \ref{lem}} \label{si:lem}

\begin{lemma}[Linearity]\label{lem}
    Under assumption \ref{ass:frame}, $Y_{ij} = \frac{\lambda_j}{\lambda_k} Y_{ik} +(\epsilon_{ij}-\frac{\lambda_j}{\lambda_k} \epsilon_{ik})$.
\end{lemma}

\begin{proof}
   Recall $Y_{ik}= \lambda_k\eta_i+\epsilon_{ik}$. Under valid measurement assumption $\lambda_k \neq 0$, $\eta_i = \frac{1}{\lambda_k} (Y_{ik}-\epsilon_{ik})$. Therefore, $Y_{ij} = \lambda_j \eta_{i}+\epsilon_{ij}=\lambda_j \frac{1}{\lambda_k} (Y_{ik}-\epsilon_{ik}) + \epsilon_{ij}= \frac{\lambda_j}{\lambda_k} Y_{ik} +(\epsilon_{ij}-\frac{\lambda_j}{\lambda_k} \epsilon_{ik})$.
\end{proof}

\subsection{Proof of Proposition \ref{prop:iv}}

\begin{proof}
    From lemma \ref{lem}, we get $Y_{ij}=\lambda_j Y_{i1} + (\epsilon_{ij}-\lambda_j \epsilon_{i1})$ by setting $k=1$.

We show (1) first. The formula suggests that $Z$ can be a valid IV. The population covariance between $Z_i$ and $Y_{ij}$ is 
$$
\begin{aligned}
Cov(Z_i,Y_{ij})&= \lambda_j Cov(Z_i,Y_{i1}) + Cov(Z_i, \epsilon_{ij}-\lambda_j\epsilon_{i1})
\end{aligned}
$$ 

The assumption that $\mathbb{E}[\eta_i^1-\eta_i^0]\neq 0$ implies that $Cov(Z_i,Y_{i1})\neq 0$. To see this, $Cov(Z_i,\eta_i^0+\tau_i Z_i + \epsilon_{i1})=\mathbb{E}\tau_i[\mathbb{E}Z_i^2-[\mathbb{E}Z_i]^2]=\mathbb{E}\tau_i Var(Z_i)$. 

For $Z_i$ to be a valid IV, it must satisfy the exclusion restriction (i.e., the assumption that the instrument has no causal pathway to the observed outcomes other than through the latent outcome itself). We can expand the covariance term and show that it is zero.
$$
\begin{aligned}
    &\;\;\;\;Cov(Z_i, \epsilon_{ij}-\lambda_2\epsilon_{i1})\\
    &=Cov(Z_i, Z_i\epsilon^1_{ij}+(1-Z_i)\epsilon^0_{ij}-\lambda_j(Z_i\epsilon^1_{i1}+(1-Z_i)\epsilon^0_{i1}))
\end{aligned}
$$

It is easy to see that each term is 0. For example, $Cov(Z_i,Z_i\epsilon^1_{ij})=\mathbb{E}[Z_i^2\epsilon^1_{ij}]-\mathbb{E}[Z_i]\mathbb{E}[\epsilon^1_{ij}]=0$ under Assumption \ref{ass:frame}, $Z_i$ is independent of potential measurement errors, and $\mathbb{E}[\epsilon^1_{ij}]=\mathbb{E}[\epsilon^0_{ij}]=0$.

Therefore, $Z$ satisfies relevance and exclusion assumptions of IV. The consistency of the IV estimator $\lambda_j=\frac{Cov(Z_i,Y_{ij})}{Cov(Z_i, Y_{i1})}$ follows immediately.

For (2), still starting from the equation $Y_{ij}=\lambda_j Y_{i1} + (\epsilon_{ij}-\lambda_j \epsilon_{i1})$, we observe that the covariance between $Y_{ij}$ and $Y_{ik}$ ($k \neq 1, k \neq j$) is 

$$
Cov(Y_{ik},Y_{ij})= \lambda_j Cov(Y_{ik},Y_{i1}) + Cov(Y_{ik}, \epsilon_{ij}-\lambda_j\epsilon_{i1})
$$

We check $Cov(Y_{ik},Y_{i1})$ first. Under assumption \ref{ass:frame} and $\{\eta^1_i,\eta^0_i\} \perp \{\epsilon^1_{ij},\epsilon^0_{ij}\}$,
\begin{align*}
    Cov(Y_{ik},Y_{i1}) &= Cov[\lambda_k (\eta^0_i + \tau_i Z_i),\eta^0_i + \tau_i Z_i]+ \lambda_k Cov(\epsilon_{ik}, \epsilon_{i1})\\
    & =\lambda_k Var[\eta^0_i + \tau_i Z_i]+0\\
    &=\lambda_k Var[\eta_i]
\end{align*}

Then, $Cov(Y_{ik},Y_{i1}) \neq 0$ if $Var[\eta_i]\neq 0$. With assumptions that $Cov(\epsilon_{ik}, \epsilon_{i1})=Cov(\epsilon_{ik}, \epsilon_{ij})=0$, $Y_{ij}$ also satisfies exclusion assumption of IV. Therefore, the IV estimator $\lambda_j=\frac{Cov(Y_{ik},Y_{ij})}{Cov(Y_{ik},Y_{i1})}$ is consistent.
\end{proof}




\subsection{Proof of Proposition \ref{prop:unbias}}

\begin{proof}

Consider (1) first.

    Note that, when $\lambda$ is known, $\mathbb{E}[\tilde{Y}^1_i]=\mathbb{E}[\eta_i^1]$ and $\mathbb{E}[\tilde{Y}^0_i]=\mathbb{E}[\eta_i^0]$.

$$
\begin{aligned}
    &\;\;\;\;\;\mathbb{E} \{\frac{1}{n_1}\sum_{i=1}^{n}Z_i\tilde{Y}_i - \frac{1}{n_0}\sum_{i=1}^{n}(1-Z_i)\tilde{Y}_i\}\\
    &=\mathbb{E} \{\frac{1}{n_1}\sum_{i=1}^{n}Z_i\tilde{Y}^1_i - \frac{1}{n_0}\sum_{i=1}^{n}(1-Z_i)\tilde{Y}^0_i\}\\
    &= \frac{1}{n_1} \sum_{i=1}^n \frac{n_1}{n}\eta^1_{i}-\frac{1}{n_0} \sum_{i=1}^n \frac{n_0}{n}\eta^0_{i}\\
    &=\frac{1}{n} \sum_{i=1}^n\mathbb{E}[\eta^1_{i}]-\frac{1}{n} \sum_{i=1}^n \mathbb{E}[\eta^0_{i}]\\
    &=\tau
\end{aligned}
$$

For (2), we need consistency of $\hat{\tau}$ when $\lambda_j$ is known first. An unbiased estimator is consistent if variance is o(1), which is clear in SI \ref{si:var}. Then, as long as $\hat{\lambda}_j$ is consistent for $\lambda_j$, the estimated $\hat{\tau}$ is consistent as well.

\end{proof}



\section{Stacking Estimator}\label{si:stack}


In practice, regression is widely used to estimate the ALTE, especially when covariates are used to improve precision. It is straightforward to show that the above difference-in-means estimator is equivalent to the ``stacked'' regression described below. 

Using rescaled outcomes, stacked regression estimates the ALTE using all outcome measures in a single pass. Let $[\tilde{Y}_i]=(
    \tilde{Y}_{i1},
    \tilde{Y}_{i2},...,
    \tilde{Y}_{iJ})^T$ denote the stacked re-scaled outcomes for individual $i$; and let $[\tilde{Y}]=([\tilde{Y}_1],[\tilde{Y}_2],....,[\tilde{Y}_n])^T$ be the vector of outcome variables $nJ \times 1$. We also stack treatment assignment so that $[Z_i]=(Z_i,Z_i,...,Z_i)^T$ is the $J \times 1$ vector and $[Z]=([Z_1],[Z_2],...,[Z_n])^T$. Then we define $X=\begin{bmatrix}
    1,1,...,1 \\
    [Z_1],[Z_2],...,[Z_n] \\
\end{bmatrix}^T$ as the $nJ \times 2$ matrix.

Consider the stacked regression where $[\tilde{Y}]$ is regressed on $X$. Then the stacking estimator $\hat{\beta}$ is the OLS coefficient for the regressor $[Z]$. 
\begin{proposition}\label{prop:stack}
    In the stacking regression where $[\tilde{Y}]$ is regressed on $X$. The OLS estimator of the coefficient of X is equivalent to $\frac{1}{n_1} Z_i \sum_{i=1}^n (\frac{1}{k}\sum_{j=1}^J \tilde{Y}_{ij}) - \frac{1}{n_0} (1-Z_i) \sum_{i=1}^n (\frac{1}{k}\sum_{j=1}^J \tilde{Y}_{ij})$.
\end{proposition}

\begin{proof}

Define the OLS estimator
$$
\begin{bmatrix}
    \hat{\alpha}\\
    \hat{\beta}
\end{bmatrix}=(X'X)^{-1}X'[\tilde{Y}]
$$

Note  that 
\begin{align}
    (X'X) &= \begin{bmatrix}
        Jn & J n_1\\
        Jn_1 & J n_1
    \end{bmatrix}
\end{align}

Let $T$ denote the index set of treated individuals, and $C$ denote the index set of control individuals. Then, algebra shows that

\begin{align}
    (X'X)^{-1}X'[\tilde{Y}] &=  \frac{1}{J^2n_1(n-n_1)}\begin{bmatrix}
        Jn_1 & -J n_1\\
        -J n_1 & J n
    \end{bmatrix} \begin{bmatrix}
        \sum_{i=1}^n\sum_{j=1}^J \tilde{Y}_{ij} \\
        \sum_{i \in T}\sum_{j=1}^J \tilde{Y}_{ij}
    \end{bmatrix}
\end{align}

Therefore, 
$$
\begin{aligned}
   \hat{\beta}
    &= \frac{1}{kn_1}\sum_{i \in T} \sum_{j=1}^J \tilde{Y}_{ij} - \frac{1}{kn_0}\sum_{i \in C} \sum_{j=1}^J \tilde{Y}_{ij}\\
    &= \frac{1}{n_1} \sum_{i \in T}(\frac{1}{k}\sum_{j=1}^J \tilde{Y}_{ij}) - \frac{1}{n_0} \sum_{i \in C} (\frac{1}{k}\sum_{j=1}^J \tilde{Y}_{ij}) \\
    &=\frac{1}{n_1} Z_i \sum_{i=1}^n (\frac{1}{k}\sum_{j=1}^J \tilde{Y}_{ij}) - \frac{1}{n_0} (1-Z_i) \sum_{i=1}^n (\frac{1}{k}\sum_{j=1}^J \tilde{Y}_{ij})
\end{aligned}
$$
\end{proof}

Recall the weighted difference-in-means estimator is $\hat{\tau} =\frac{1}{n_1}\sum_{i=1}^{n}Z_i [\sum_{j=1}^J \omega_j\tilde{Y}_{ij}]- \frac{1}{n_0}\sum_{i=1}^{n}(1-Z_i)[\sum_{j=1}^J \omega_j \tilde{Y}_{ij}]$. We can see the stacking regression estimator is equivalent to the weighted difference-in-means estimator when $\omega_j = \frac{1}{k}$. To achieve an optimal weighting scheme, we can simply multiply the optimal weight $\omega_j^*$ by $\tilde{Y}_{ij}$ before stacking, i.e. $[\tilde{Y}_i]=(
    \omega_1^*\tilde{Y}_{i1},...,
    \omega_J^*\tilde{Y}_{iJ})^T$. If the researcher seeks to adjust for covariates, they can simply be added to the regression specification.  Covariate-adjusted stacked regression provides a consistent estimator of the average treatment effect on the latent outcome, and its variance may be smaller than the unadjusted stacked regression when the covariates are prognostic of the observed outcomes.

\section{Variance of Difference-in-means based on a Weighted Scaled Index}\label{si:var}

Our estimator has a similar flavor to the IPW estimator. When $\lambda_j$ scaling factors are known based on the measurement properties observed in prior studies, and weight $\omega$ is pre-specified, then it is straightforward to work out the variance of this estimator. Instead, when it is estimated, the standard approach is to use GMM.

\subsection{Scenario 1: when $\lambda$ and $\omega$ are known}
In this scenario, we can treat each $\hat{\lambda}_j=\lambda_j$ as a constant and ignore its uncertainty, in which case the variance is similar to the Neyman variance. 

We define the final weighted outcome as follows: $\tilde{Y}^1_i := \sum_{j=1}^k \frac{\omega_j \lambda_j}{\hat{\lambda}_j}[\eta_i^1 + \epsilon'^1_{ij}]$ and $\tilde{Y}^0_i := \sum_{j=1}^k \frac{\omega_j \lambda_j}{\hat{\lambda}_j} [\eta_i^0 + \epsilon'^0_{ij}]$, where $\epsilon'^z_{ij}$ is the scale potential measurement error $\frac{1}{\lambda}_j\epsilon^z_{ij}$. With the new notation $\tilde{Y}^1_i$ and $\tilde{Y}^0_i$, we define the observed $\tilde{Y}_i=Z_i\tilde{Y}^1_i+(1-Z_i)\tilde{Y}^0_i$. Then standard argument of population variance applies, for example see \citet[][Chapter 9]{ding2024first}.

Therefore, the variance of the weighted average estimator is $\frac{\sigma_1^2}{n_1}+\frac{\sigma_0^2}{n_0}$, where $\sigma_1^2 = Var_{sp}[\tilde{Y}^1_i]=\mathbb{E}_{sp}[(\tilde{Y}^1_i-\mathbb{E}_{sp}[\tilde{Y}^1_i])^2]$ and $\sigma_0^2 = Var_{sp}[\tilde{Y}^0_i]$. The variance estimator can also be estimated using sample analogs: $\frac{\hat{S}_1^2}{n_1}+\frac{\hat{S}_0^2}{n_1}$, where $\hat{S}_1^2=\frac{1}{n-1}\sum_{i=1}^n Z_i(\tilde{Y}_i- \frac{1}{n_1} \sum_{i=1}^n Z_i\tilde{Y}_i)^2$ and $\hat{S}_0^2=\frac{1}{n-1}\sum_{i=1}^n (1-Z_i)(\tilde{Y}_i- \frac{1}{n_1} \sum_{i=1}^n (1-Z_i)\tilde{Y}_i)^2$.

\subsection{Scenario 2: when $\lambda$ and/or $\omega$ are estimated}\label{si:gmm}

Similar to IPW estimator, we can use GMM to incorporate the uncertainty of $\lambda$ estimation. Consider the example with three measures and pooled measurement errors. We seek to estimate $\lambda_2$, $\lambda_3$, and $\tau$. Because $\lambda_2$, $\lambda_3$ are over-identified through multiple IVs, as we shown in Proposition \ref{prop:iv}, the moment conditions are:

\begin{align*}
    \mathbb{E} \begin{bmatrix}
         Y_{i2}-\alpha_0- \lambda_2 Y_{i1} \\
        (Y_{i2}-\alpha_0- \lambda_2 Y_{i1})Z_i\\
        (Y_{i2}-\alpha_0- \lambda_2 Y_{i1})Y_{i3} \\
        Y_{i3}-\alpha_1- \lambda_3 Y_{i1} \\
        (Y_{i3}-\alpha_1- \lambda_3 Y_{i1})Z_i\\
        (Y_{i3}-\alpha_1- \lambda_3 Y_{i1})Y_{i2} \\
(\frac{Z_i}{p} - \frac{1-Z_i}{1-p})\overline{Y}_i-\tau
    \end{bmatrix} =0
\end{align*}

The first three equations are over-identified IV equations that estimate $\lambda_2$. $\overline{Y}_i=\sum_{j=1}^J \frac{1}{\lambda_j} Y_{ij}$ is the weighted average. The last equation is the weighted-average estimator for the average treatment effect, where $p$ is the propensity score. Then, the variance of $\tau$ is estimated through traditional sandwich estimators that incorporate the uncertainty in the estimate of $\lambda_j$.
The last two rows can be replaced by $ [( Y_{i1} + \frac{Y_{i2}}{\lambda_2}+\frac{Y_{i3}}{\lambda_3})/3 - \mu_1]Z_i$ and $[( Y_{i1} + \frac{Y_{i2}}{\lambda_2}+\frac{Y_{i3}}{\lambda_3})/3 - \mu_0](1-Z_i)$, so that the difference-in-means estimator is expressed in two parts rather than one: $\hat{\tau}=\hat{\mu}_1-\hat{\mu}_0$ and the variance is $Var(\mu_1)+Var(\mu_0)-2Cov(\mu_1,\mu_0)$.

For an optimally weighted scaled index, because we need to deal with covariance, we centered the outcome measures first. For ease of exposition, we assume the variance of potential measurement is the same in the treatment and control groups. Then, we can use pooled estimation of the variance of the measurement error. The following $Y_{ij}$ is assumed to be demeaned first.

\begin{align*}
    \mathbb{E} \begin{bmatrix}
        (Y_{i2}- \lambda_2 Y_{i1})Z_i\\
        (Y_{i2}- \lambda_2 Y_{i1})Y_{i3} \\
        (Y_{i3}- \lambda_3 Y_{i1})Z_i\\
        (Y_{i3}- \lambda_3 Y_{i1})Y_{i2} \\
        Y_{1i}Y_{2i} - \lambda_2 \psi \\
        Y^2_{1i}-\psi-\sigma^2_1 \\
        Y^2_{2i}-\lambda^2_2\psi-\sigma^2_2 \\
        Y^2_{3i}-\lambda^2_3\psi-\sigma^2_3 \\
        (\frac{Z_i}{p} - \frac{1-Z_i}{1-p})Y^{opt}_i-\tau
    \end{bmatrix} =0
\end{align*}

In the above moment conditions, $\psi$ is the variance of $\eta$, $\sigma^2_j$ is the variance of $\epsilon_{\cdot j}$, $Y^{opt}_i=\sum_{j=1}^J \omega_j Y_{ij}$, and $\omega_j = \frac{\lambda^2_j/\sigma^2_j}{\sum_{j=1}^k \lambda^2_j/\sigma^2_j}$. For the more general case where the variances differ between the two groups, we simply add additional moment equations for each group separately. 

\newpage






\section{Comparison to current practice}\label{si:comp1}

\subsection{Comparison to PCA, ICW, and SEM}
In social science, it is quite common for a study's key outcome variable to be an abstract concept that is not directly observed. To better measure this latent variable, researchers often seek to measure multiple observable manifestations of this latent outcome. 
However, researchers are aware that critics may express concern about false discovery due to multiple comparisons. To sidestep this concern, researchers often attempt to reduce the dimensionality of their outcomes. 

Common dimensionality reduction techniques include summary index creation, principal components analysis, or other techniques that in some way standardize the latent and/or observed outcomes. A key feature of those methods is that they extract or construct a low-dimensional, typically single-dimensional, variable.  Due to standardization (e.g., recoding each outcome measure to have a variance of 1), these methods generate results that are sample-specific in the sense that the scaling used to define the outcome variable depends on the dispersion of scores among the subjects at hand. Even if the true causal relationship between treatment and outcome were identical in two different (large) samples, estimated effects might differ substantially if the dispersion of unobservables differs.  
As noted in the main text, our method solves the above problem by setting the metric of the latent variable to be the same as one of the observed variables. This ``unstandardized approach'' allows us to interpret the latent variable with the same metric as that observed outcome variable (see \citet[pp. 239-240]{bollen1989structural} and \citet[pp. 28-29]{loehlin1998latent}). The results are comparable across samples even when the variances of the outcome variables differ.

Another approach, as noted in the main text, is to simply consider observed outcomes as distinct from one another, even if they are in the same substantive domain.  There is no latent variable in this case; there are instead distinct multiple outcomes.  In his widely-cited article
\citet{anderson2008multiple} seeks to find a weight vector $w \in R^{J}$ so that the variance of the weighted outcome $wY \in R^{n \times 1}$ is minimized, subject to the constraint that $1'w=1$. This is an optimization problem, where the objective function is $\min_{w} w'\Sigma w$. By using the Lagrange multiplier method, one obtains the weight $w=(1'\Sigma^{-1}1)^{-1}(\Sigma^{-1}1)$. 

How well does Anderson's inverse-covariance weighting (ICW) method perform in terms of statistical power?  (We focus on power because we know that Anderson's standardized index will not yield unbiased or consistent estimates of the ALTE in unstandardized units). 


The result of Figure \ref{si:powermain} is based on the following simulation:
 we set $\eta_{i0}=0$, $\eta_{i1}=N(\theta,1)$, where $\theta \in \{0,0.05,0.15,0.25,0.35,0.45\}$. For three outcome measures, we let $\lambda_{j}=1$ and measurement errors $e_1=N(0,0.5)$, $e_2=N(0,0.1)$, and $e_3=N(0,2)$. 



The next simulation compares ICW, PCA and SUR (based on the F-test). We generate data from the DGP in Figure \ref{si:dgpsur1}. We examine two situations, low correlation (0.2) and high correlation (0.8) among the disturbances.  ICW and PCA are somewhat more powerful than SUR in both scenarios.

\begin{figure}[!h]
    \centering
    \includegraphics[width=0.9\linewidth]{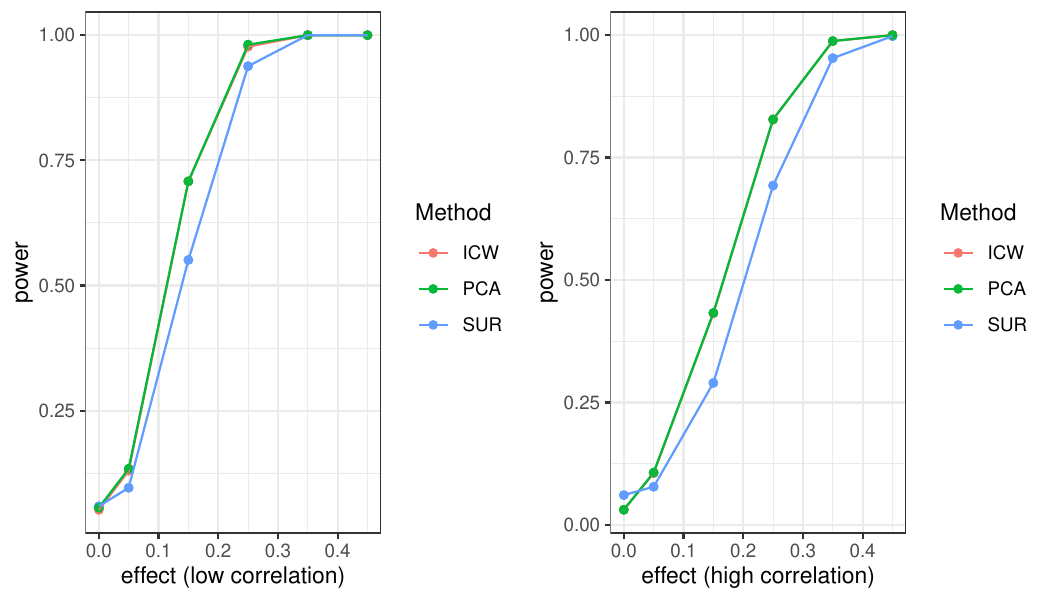}
    \caption{Power analysis: SUR, ICW, and PCA}
    \label{si:dgpsur1}
\end{figure}

\subsection{Comparison to Hierarchical Item Response Theory} \label{si:comp2}


\citet{stoetzer2022causal} propose a hierarchical item response theory (hIRT) model to estimate the latent treatment effect. Their model has two components. First, the latent variable $\eta_i$ is assumed to be a linear function of the treatment: $\eta_i = \gamma_0 + \gamma_1 Z_i + \epsilon_i$, where $\epsilon_i$ is normally distributed with a constant variance $\sigma^2$. In this equation, $\gamma_1$ is the average treatment effect of interest. Next, the observed outcomes $Y_{ij}$ are generated by a characteristic function of the item: $\mathbb{P}[Y_{ij}=h | \eta_i]=P_{ijh}(\eta_i;\alpha_{ijh},\beta_{ij})$, where $\alpha_{ijh}$ and $\beta_{ij}$ are the item difficulty and item discrimination parameters for item $j$ and answering $h$. 

As \citet{stoetzer2022causal} note (e.g. page 28), this model makes a set of parametric assumptions, including a constant treatment effect and a particular item characteristic function. Our approach relaxes some of these assumptions. We allow $Z_i$ to have heterogeneous treatment effects. We also remain agnostic about the non-linear transformation. 

To estimate $\gamma_1$ under the hIRT, several other identification restrictions must be imposed. For example, $\lambda_0=0$ and $\sigma^2=1$. Then, $\gamma_1$ can be estimated by the EM algorithm. These technical assumptions serve the same normalization purpose as setting $\lambda_1=1$. However, setting $\lambda_1=1$ allows us to give a natural interpretation to the latent variable.

An immediate shortcoming of the parametric IRT approach is that the estimator is likely to be inconsistent when the parametric functional form is misspecified or there are heterogeneous effects. To illustrate the problem, we conduct two simulations. For the constant effect simulation, we generate potential outcomes $\eta_{i0}=N(0,1)$ and $\eta_{i1}= 2+ \eta_{i0}$ so that treatment effect is exactly 2 for every individual. For heterogeneous effects, we simply add some mean zero variation to generate treated potential outcomes: $\eta'_{i1}= 2 + \eta_{i0} + N(0,1)$. Note that the average treatment effect is still maintained at 2. Next, we generate five measures (each with 13 levels) using an IRT model. Figure \ref{fig:irt} illustrates the mean estimate and its 95\% confidence interval. When the effect of treatment is constant for every individual, hIRT recovers the true effect of 2. However, after introducing even mild heterogeneity in treatment effects, hIRT exhibits substantial bias.
\begin{figure}[!h]
    \centering
    \includegraphics[width=0.7\linewidth]{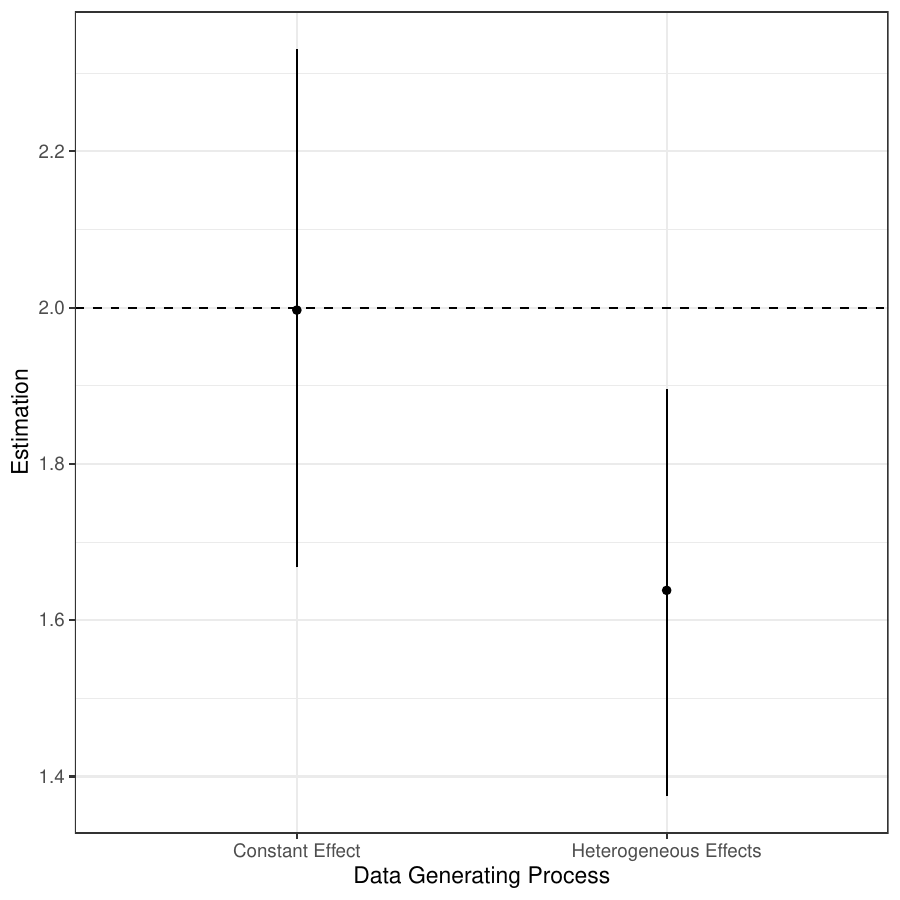}
    \caption{The Distribution of hIRT Estimates under Constant Treatment Effects and Heterogeneous Treatment Effects.}
    \label{fig:irt}
\end{figure}

\subsection{Summary}

Table \ref{tab:comparison} provides an overview of several commonly used methods for scoring and analyzing outcomes.  When key assumptions are met -- most importantly, that the measured outcomes are linear manifestations of a latent outcome -- optimal scaling and structural equation modeling produce efficient estimates. SEM has the advantage of facilitating over-identification tests, while optimal scaling has the advantage of facilitating visualizations of the estimated ALTE using individual-level data. The downside of the linearity assumption is the burden that it places on data collection; researchers must gather granular outcome measures, perhaps by combining a series of binary outcome measures to create an additive index. Hierarchical IRT has the advantage of leveraging sparse outcome measurement but imposes strong assumptions, such as constant treatment effects across units.  The use of principal components analysis to form a summary index may be useful for hypothesis testing but does not render point estimates in a readily interpretable or transportable metric. The same is true for an inverse covariance weighted index, and when measured outcomes are manifestations of the same latent outcome, ICW tends to be less efficient.  The most agnostic method is seemingly unrelated regression, which does not presuppose that the outcome measures tap the same latent variable. This model provides a useful null model against with to judge the adequacy of the SEM model. 

\afterpage{%
    \clearpage
    \thispagestyle{empty}
    \singlespacing
    \begin{landscape}
        \centering 
    \captionof{table}{Comparison of Methods for Analyzing Multiple Outcomes}\label{tab:comparison}
     \begin{threeparttable}
       \begin{tabular}{llp{70pt}p{50pt}p{100pt}p{100pt}} 
\\[-1.8ex]\hline 
\hline 
{\textbf{Method}} & {\textbf{One-step?}$^a$} & {\textbf{Estimand}$^b$}   & {\textbf{Allows HTE?}$^c$}   & {\textbf{Advantages}}    & {\textbf{Disadvantages}}     \\
\hline 
Optimal Weighted Average Index & No & ALTE (unstandardized) & Yes  & Efficient when assumptions hold; facilitates visualization &  Requires multiple outcome measures in order to satisfy linearity assumption \\
\hline 
Structural Equation Model (SEM) &Yes & ALTE (unstandardized) &Yes  &Efficient when assumptions hold; enables nested model comparisons &  Requires multiple outcome measures in order to satisfy linearity and assumes multivariate normality \\
\hline 
Hierarchical Item Response Theory  & Yes & ALTE (scaled) & No &Requires few outcome measures &  Requires constant treatment effects and homoskedasticity\\
\hline 
Principal Components Analysis  & No  &  ATE of first principal component  & Yes  & Allows hypothesis testing without a latent variable model              &  Estimated ATE has a sample-specific scale  \\
\hline 
Inverse Covariance Weights &No & Reweighted ATE for each outcome & Yes &  Allows hypothesis testing without a latent variable model& Estimated ATE has a sample-specific scale  \\
\hline 
Seemingly Unrelated Regression & Yes & ATE for each outcome & Yes  &  Allows for hypothesis testing without latent variable model   & 
Less efficient than SEM when outcomes are linear manifestations of a latent variable\\
\hline 
\hline 
\end{tabular}\begin{tablenotes}
\footnotesize
\item \textsuperscript{a} Two-step estimators produce biased variance estimates.  Generalized method of moments may be used to generate one-step estimates using an optimally weighted index.
\textsuperscript{b}  For details on scaling and weighting, see SI \ref{se:alg}.  ALTE = Average latent treatment effect. 
\textsuperscript{c}  HTE = heteogeneous treatment effects.
\end{tablenotes}
\end{threeparttable}
    \end{landscape}
    \clearpage
}

\subsection{Pseudo-algorithms to Generate Outcome Indices}\label{se:alg}

\subsubsection{Difference-in-means based on a Weighted Scaled Index}

\begin{enumerate}
    \item \textbf{Data:} $Z_i$, $Y_{ij}$ ($i=1,2,..n; j = 1,2,...,K$), $n_1=\sum_{i=1}^n Z_i$, $n_0=n-n_1$
    \item \textbf{Estimate $\lambda_j$}: set $\lambda_1 = 1$, for $\lambda_j \neq \lambda_1$, use either of the following estimators:
        \begin{enumerate}
            \item [(1)] SEM: $\eta =\sim 1*Y_1 + Y_2 + ... + Y_n$; or
            \item [(2)] GMM: using $Z_i$ or/and $Y_{ik} \forall k \neq 1$ and $k \neq j$ as IVs
        \end{enumerate}
    \item \textbf{Scale outcomes:} for each $i,j$, $\tilde{Y}_{ij} \leftarrow \frac{Y_{ij}}{\hat{\lambda}_j}$
    \item \textbf{Weight outcomes:} $\tilde{Y}_i \leftarrow \sum_{j=1}^K \omega_j \tilde{Y}_{ij}$, where $\sum_{i=1}^K \omega_j = 1$. Common options:
        \begin{enumerate}
            \item [(1)] Uniform weights: $\omega_j \leftarrow \frac{1}{K}$

            \item [(2)] Optimal weights: $\omega_j \leftarrow \frac{\hat{\lambda}^2_j/\hat{\sigma}^2(\epsilon_j)}{\sum_{j=1}^k \hat{\lambda}^2_j/\hat{\sigma}^2(\epsilon_j)}$
        \end{enumerate}
    \item \textbf{Output:} $\hat{\tau} \leftarrow  \frac{1}{n_1}\sum_{i=1}^{n}Z_i\tilde{Y}_i - \frac{1}{n_0}\sum_{i=1}^{n}(1-Z_i)\tilde{Y}_i$ 
\end{enumerate}

\subsubsection{Anderson's ICW}

\begin{enumerate}
    \item \textbf{Data:} $Z_i$, $Y_{ij}$ ($i=1,2,..n; j = 1,2,...,K$), $n_1=\sum_{i=1}^n Z_i$, $n_0=n-n_1$
    \item \textbf{Standardize:} For each $j$, $\tilde{Y}_{ij} = \frac{(Y_{ij}-\overline{Y}_{ij})}{sd(Y^c_{ij})}$, where $sd(Y^c_{ij})$ is  control group standard deviation.
    \item \textbf{Weight outcomes:} $S_{i} \leftarrow \ \omega_j \tilde{Y}_{ij}$, where $\omega_j=(1'\Sigma^{-1}1)^{-1}(1'\Sigma)$ is the weight. Inverted covariance matrix $\Sigma^{-1}$ can be calculated using function cov.wt in R.
    \item \textbf{Output:} $\hat{\tau} \leftarrow  \frac{1}{n_1}\sum_{i=1}^{n}Z_i S_{i} - \frac{1}{n_0}\sum_{i=1}^{n}(1-Z_i)S_{i}$ 
\end{enumerate}

\subsubsection{hIRT}

\begin{enumerate}
    \item \textbf{Data:} $Z_i$, $Y_{ij}$ ($i=1,2,..n; j = 1,2,...,K$)

    \item \textbf{Model:} $\mathbb{P}[Y_{ij}=h|\eta_i]=\frac{exp(\alpha_{jh}+\beta \eta_i)}{1+exp(\alpha_{jh+\beta \eta_i})}-\frac{exp(\alpha_{jh+1}+\beta \eta_i)}{1+exp(\alpha_{jh+1}+\beta \eta_i)}$ and $\eta_i=\gamma_0 + \gamma_1 Z_i + \epsilon_i$
    
    \item \textbf{Estimation:} The EM algorithm gives final estimates for each parameter; in R, one can use function hgrm of the hIRT package.
    
    \item \textbf{Output:} $\hat{\tau} \leftarrow \hat{\gamma}_1$
\end{enumerate}

    
    


\subsubsection{PCA}

\begin{enumerate}

\item \textbf{Data:} $Z_i$, $Y=[Y_{ij}]$ ($i=1,2,..n; j = 1,2,...,K$), $n_1=\sum_{i=1}^n Z_i$, $n_0=n-n_1$

\item \textbf{Demean:} $\overline{Y}$ be the item-wise zero empirical mean outcome matrix. 

\item \textbf{Weight outcomes (first component)}. $\tilde{Y} \leftarrow \omega \overline{Y}$. The weight $\omega$ maximize the variance of $\tilde{Y}$ :  $\omega=argmax_{||\omega||=1} || \overline{Y}\omega||^2=argmax \frac{\omega^TY^TY\omega}{\omega^T\omega}$.

\item \textbf{Output:} $\hat{\tau} \leftarrow  \frac{1}{n_1}\sum_{i=1}^{n}Z_i \tilde{Y}_i - \frac{1}{n_0}\sum_{i=1}^{n}(1-Z_i)\tilde{Y}_i$ 
\end{enumerate}

\section{Illustration of How Additional Items Reduce the Sampling Variability of the Estimated ALTE}\label{si:variance}

\subsection{Optimal weighting never increases the variance}
When researchers use equal weights to construct an additive index of outcome measures, the variance of $\tilde{Y}^1_i$ is $\frac{1}{n_1}\sigma^2(\eta^1_i) + \frac{1}{n_0}\frac{1}{J^2}\sum_{j=1}^J \frac{\sigma^2(\epsilon_{\cdot j})}{\lambda^2_j}$. The gain from adding one more measure is $$
\frac{1}{(J+1)^2}\frac{\sigma^2(\epsilon_{\cdot J+1})}{\lambda^2_{J+1}} - (\frac{1}{J^2}-\frac{1}{(J+1)^2})\sum_{j=1}^J \frac{\sigma^2(\epsilon_{\cdot j})}{\lambda^2_j}
$$ Suppose the new outcome measure is quite noisy: the variance of the measurement error $\sigma^2(\epsilon_{\cdot J+1})$ is large. Therefore, the first term can dominate the second term in the formula, which makes the variance increase rather than decrease. For this reason, those who advocate the creation of simple additive indexes caution that there is no guarantee that including an additional measure will improve reliability on the margin \citep[p.218]{ansolabehere2008strength}.

\subsection{The Gains from Adding Measures to a Weighted Scaled Index} \label{si:measure}

A simple scenario illustrates the gains from including an additional measure and applying optimal weights. 
For ease of exposition, we assume a balanced design ($n_0 = n_1$) 
and consider the case in which all outcome measures have the same error variance $\sigma^2(\epsilon)$ and factor loading $\lambda$. Let $\Sigma:=\frac{\lambda^2}{\sigma^2(\epsilon)}$. In this case, where all measures are equally reliable, the inclusion of an additional measure to the current $J$ measures can be expected to change the variance of the estimated ALTE by

$$
\Delta(J):=\frac{4}{n}\frac{-1}{J(J+1)\Sigma}
$$


Due to the rapidly increasing quadratic term $J(J+1)$ in the denominator, the payoff diminishes with each additional measurement. Moreover, the precision gain from the $(J+1)^{th}$ measure depends on $\Sigma$, inverse of the variance of measurement error, or the precision of the measure. As the measurement error variance increases, the precision improvement from an additional measure increases. 

We can get a better feel for the advantages of including an additional outcome measure by examining the percentage change. We compare the reduction in variance from adding the $(J+2)^{th}$ measure with the reduction in variance from adding the $(J+1)^{th}$ measure. 


$$
\frac{\Delta(J+1)}{\Delta(J)}=\frac{J}{(J+2)}
$$

For example, suppose that we have only one outcome measure, $J=1$. The denominator is the variance reduction by adding the second measure. The numerator is the variance reduction of the third measure. The gain is $\frac{1}{3} \approx 33.33\%$. 
If we add a fourth measure, the variance reduction will be $\frac{1}{6}=16.67\%$ compared to the gain from the second measure. From this example, we surmise that when outcome measures are equally reliable, three or four measures are sufficient in practice, although more measures may be helpful for other reasons, such as linearizing the relationship between $\eta$ and sub-indices of outcome measures.

By way of illustration, we generate $n=500$ simulated potential outcomes $\eta^1 \sim N(1,1)$ and $\eta^0 \sim N(2,1)$, and 6 outcome measures $Y_{ij}=Z\eta^1_i + (1-Z_i)\eta^0_i + \epsilon_i$. We consider two cases, high reliability and low reliability, brought about by changing the variance of the measurement error. The average of 1000 simulation results is shown in Figure \ref{fig:red} and Table \ref{tab:redt}. Generally, more measures decrease the estimation variance. The benefits of additional measures depend on the measures' reliability. When measures are already quite reliable, the variance reduction is small. 

\begin{figure}[!h]
    \centering
    \includegraphics[width=0.8\linewidth]{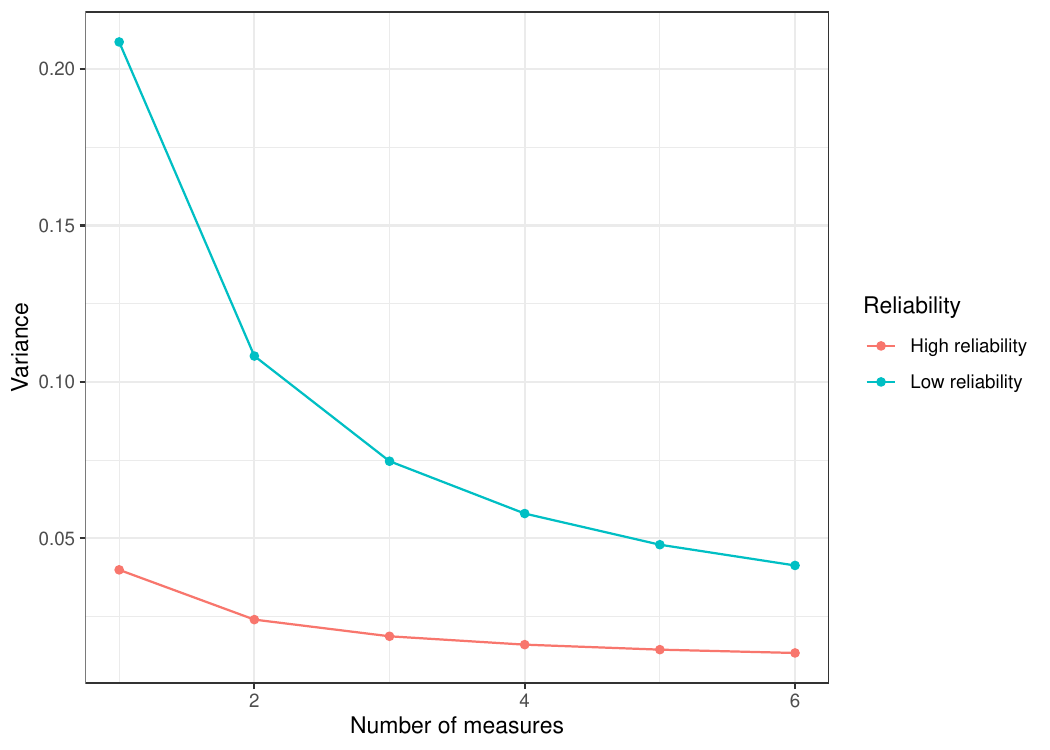}
    \caption{\textbf{Simulation: Variance reduction and Reliability.} The horizontal line represents the number of measures and the vertical line is the estimated variance of the optimal weighting estimator. The variance reduction is larger if the reliability is lower.}
    \label{fig:red}
\end{figure}

The first column in Table \ref{tab:redt} indicates the number of measures and the third column is the estimated variance for the ALTE. The fourth column calculates the $\Delta(J)$. Recall $\Delta(J)$ is the variance reduction by adding the $J(+1)^{th}$ measure. The last two columns compare the simulation ratio and the theoretical value $\frac{J}{J+2}$ we derived before.

\begin{table}[ht]
\centering
\caption{Simulation: Variance reduction and Reliability}\label{tab:redt}
\begin{tabular}{cccccc}
  \hline
  \hline
  Number & Reliability & Var & $\Delta(J)$ & $\Delta(J+1)/\Delta(J)$ & Theory \\ 
  \hline
   1 & High  & 0.0399 & -0.0159 & 0.3351 & 0.3333 \\ 
     2 & High  & 0.0240 & -0.0053 & 0.5011 & 0.5000 \\ 
    3 & High  & 0.0187 & -0.0027 & 0.5985 & 0.6000 \\ 
    4 & High  & 0.0160 & -0.0016 & 0.6645 & 0.6667 \\ 
     5 & High  & 0.0144 & -0.0011 &  &  \\ 
    6 & High  & 0.0133 &  &  &  \\ 
    1 & Low  & 0.2087 & -0.1004 & 0.3350 & 0.3333 \\ 
    2 & Low  & 0.1083 & -0.0336 & 0.4973 & 0.5000 \\ 
     3 & Low  & 0.0747 & -0.0167 & 0.5957 & 0.6000 \\ 
     4 & Low  & 0.0579 & -0.0100 & 0.6657 & 0.6667 \\ 
    5 & Low  & 0.0480 & -0.0066 &  &  \\ 
    6 & Low  & 0.0413 &  &  &  \\ 
   \hline
   \hline
\end{tabular}
\end{table}

\section{More Information on the Empirical Application}\label{si:app}

\subsection{Data}

Here we list the outcome measurements and covariates we used in the application section.

\textbf{Anti-Immigrant Prejudice Index.} The first set of questions are five point scales where respondents were asked: "Do you agree or disagree with the below statements about undocumented or illegal immigrants?" Response options were: Strongly agree, Somewhat agree, Neither agree nor disagree, Somewhat disagree, Strongly disagree:
\begin{itemize}
    \item `living': "I would have no problem living in areas where undocumented immigrants live."
    \item `fit': "Too many undocumented immigrants just don’t want to fit into American society."
    \item `burden': "Undocumented immigrants are too much of a burden on our communities."
    \item `crime': "Undocumented immigrants have already broken the law coming here illegally, so they are more likely to commit other crimes."
    \item `values': "Undocumented immigrants hold the same values as me and my family."

\end{itemize}

\textbf{Anti-Immigrant Policy Index.} Respondents were first asked: "Politicians are considering a number of policies about immigration. We want to know what you think. Do you agree or disagree with the statements below?" Response options were: Strongly agree, Somewhat agree, Neither agree nor disagree, Somewhat disagree, Strongly disagree:

\begin{itemize}
    \item `daca': "The federal government should grant legal status to people who were brought to the US illegally as children and who have graduated from a U.S. high school."
    \item `citizenship': "The federal government should allow undocumented immigrants currently in the U.S. to become citizens after they have lived, worked, and paid taxes for at least 5 years."
    \item `compassion': "Undocumented immigrants deserve compassion and should not live in daily fear of deportation."
\end{itemize}


\textbf{Baseline covariates.} They include the pre-treatment variables `daca',`citizenship',`living' and `fit'. We add them together to construct a single index.

\subsection{Supplementary results}

\begin{table}[!htbp] \centering \renewcommand*{\arraystretch}{1.1}\caption{Summary Statistics}\label{tab:summary}
\resizebox{\textwidth}{!}{
\begin{tabular}{lrrrrrrr}
\hline
\hline
Variable & N & Mean & Std. Dev. & Min & Pctl. 25 & Pctl. 75 & Max \\ 
\hline
Full Treatment & 7870 & 0.33 & 0.47 & 0 & 0 & 1 & 1 \\ 
Abbreviated Treatment & 7870 & 0.33 & 0.47 & 0 & 0 & 1 & 1 \\ 
daca & 1578 & 0.6 & 1.4 & -2 & -1 & 2 & 2 \\ 
citizenship & 1578 & -0.14 & 1.5 & -2 & -2 & 1 & 2 \\ 
compassion & 1578 & -0.34 & 1.5 & -2 & -2 & 1 & 2 \\ 
living & 1578 & 0.51 & 1.3 & -2 & 0 & 2 & 2 \\ 
values & 1578 & 0.58 & 1.2 & -2 & 0 & 2 & 2 \\ 
fit & 1578 & -0.4 & 1.4 & -2 & -2 & 1 & 2 \\ 
burden & 1578 & -0.27 & 1.4 & -2 & -2 & 1 & 2 \\ 
crime & 1578 & -0.87 & 1.2 & -2 & -2 & 0 & 2 \\ 
\hline
\hline
\end{tabular}
}
\end{table}

\begin{table}[ht]
\centering
\caption{Covariance Matrix}\label{si:cov}
\resizebox{\textwidth}{!}{
\begin{tabular}{rrrrrr}
  \hline
  \hline
 & Treatment (full) & Treatment (mod) & Attitudes & Policy Views & Covariates \\ 
  \hline
Treatment (full) (Ave: 0.33) & 0.226 & -0.109 & 0.061 & 0.104 & -0.035 \\ 
  Treatment (mod) (Ave: 0.33) & -0.109 & 0.216 & 0.008 & 0.010 & 0.052 \\ 
  Attitudes (Ave: 2.06) & 0.061 & 0.008 & 12.328 & 15.187 & 12.775 \\ 
  Policy Views (Ave: 2.63) & 0.104 & 0.010 & 15.187 & 30.222 & 19.782 \\ 
  Covariates (Ave: 1.76) & -0.035 & 0.052 & 12.775 & 19.782 & 19.318 \\ 
   \hline
   \hline
\end{tabular}
}
\end{table}

\begin{figure}[!h]
    \centering
    \includegraphics[width=0.9\linewidth]{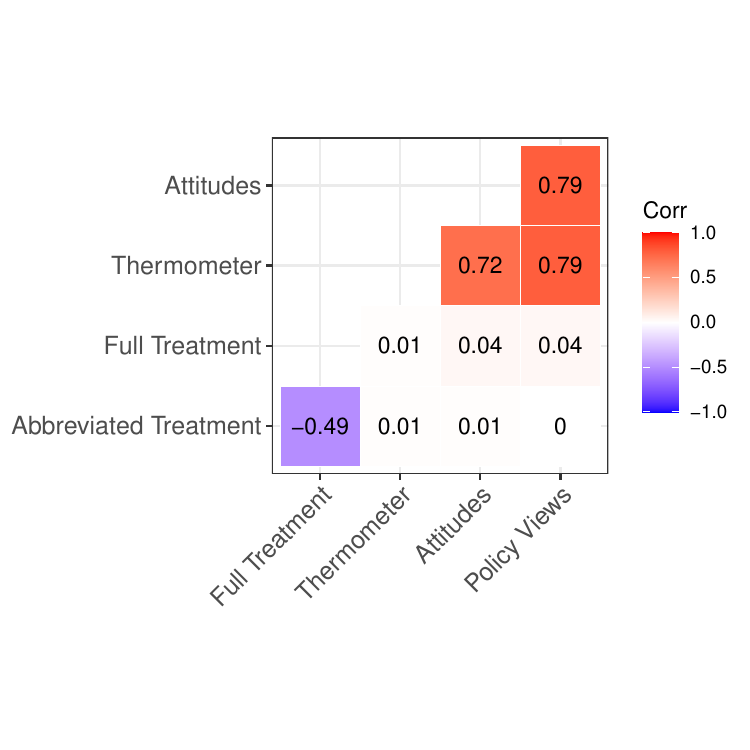}
    \caption{Correlation Matrix}
    \label{fig:cor}
\end{figure}

\begin{figure}[!h]
    \centering
    \includegraphics[width=0.7\linewidth]{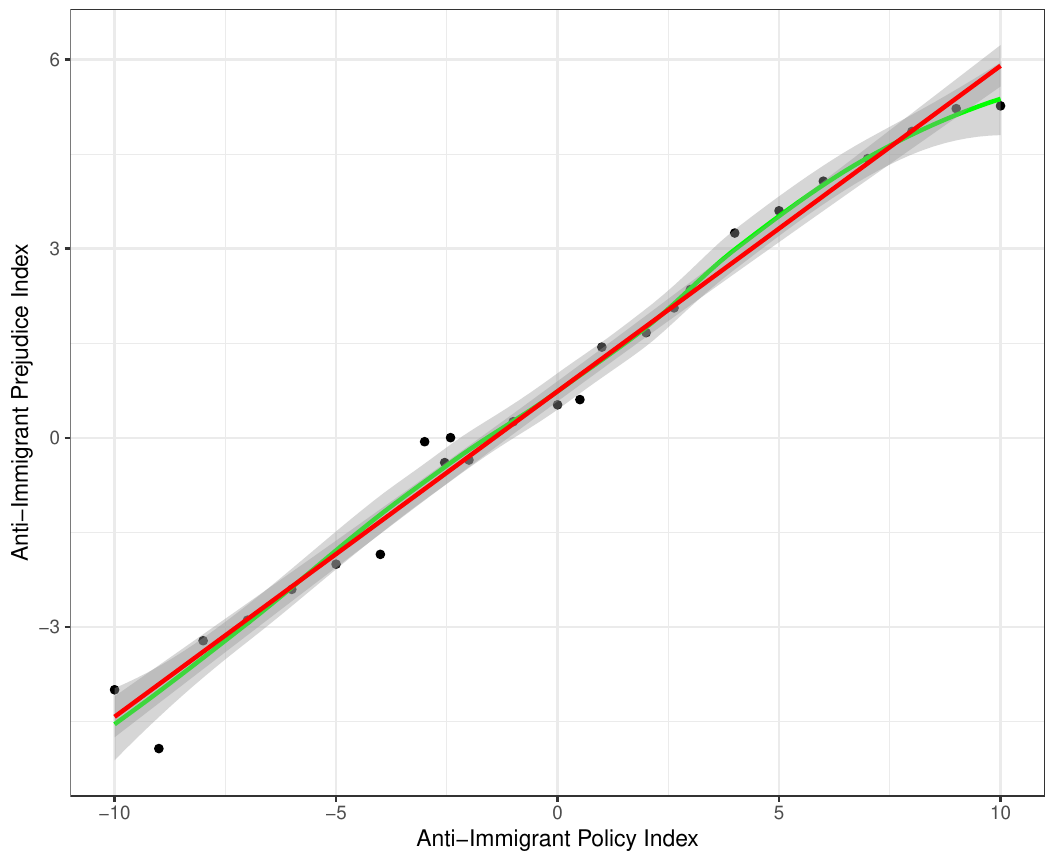}
    \caption{Linearity check}
    \label{fig:lin}
\end{figure}


\begin{table}[ht]
\caption{Variance Estimation}\label{tab:appvar}
\centering
\begin{tabular}{rlrrll}
  \hline
  \hline
 & \multicolumn{4}{c}{Model without Baseline covariates} \\
 \hline
  Variables & Var est & se & z & pvalue \\ 
  \hline
.Attitudes  & 3.14 & 3.24 & 0.97 & 0.33 \\ 
 .Policy Views  & 5.11 & 8.86 & 0.58 & 0.56 \\ 
   Treat (full)   & 0.23 & 0.01 & 28.09 & 0.00 \\ 
   Treat (mod)  & 0.22 & 0.01 & 28.09 & 0.00 \\ 
   .eta & 9.16 & 3.26 & 2.81 & 0.00 \\ 
   \hline
\end{tabular}
\begin{tabular}{rlrrll}
 & \multicolumn{4}{c}{Model with Baseline covariates} \\
  \hline
 .Attitudes & 2.52 & 0.13 & 18.97 & 0.00 \\ 
 .Policy Views & 6.70 & 0.33 & 20.00 & 0.00 \\ 
Treat (full) & 0.23 & 0.01 & 28.09 & 0.00 \\ 
 Treat (mod) & 0.22 & 0.01 & 28.09 & 0.00 \\ 
 Cov & 19.31 & 0.69 & 28.09 & 0.00 \\ 
 .eta & 1.32 & 0.11 & 12.17 & 0.00 \\
   \hline
   \hline
\end{tabular}
\begin{tablenotes}     
\item {\it Note}: This table shows the variance estimation of observed and latent variables in the application. The $\cdot$ in front of the parameter denotes residual variance for the variable. 
 \end{tablenotes}
\end{table}

\begin{figure}[!h]
    \centering
\includegraphics[width=0.9\linewidth]{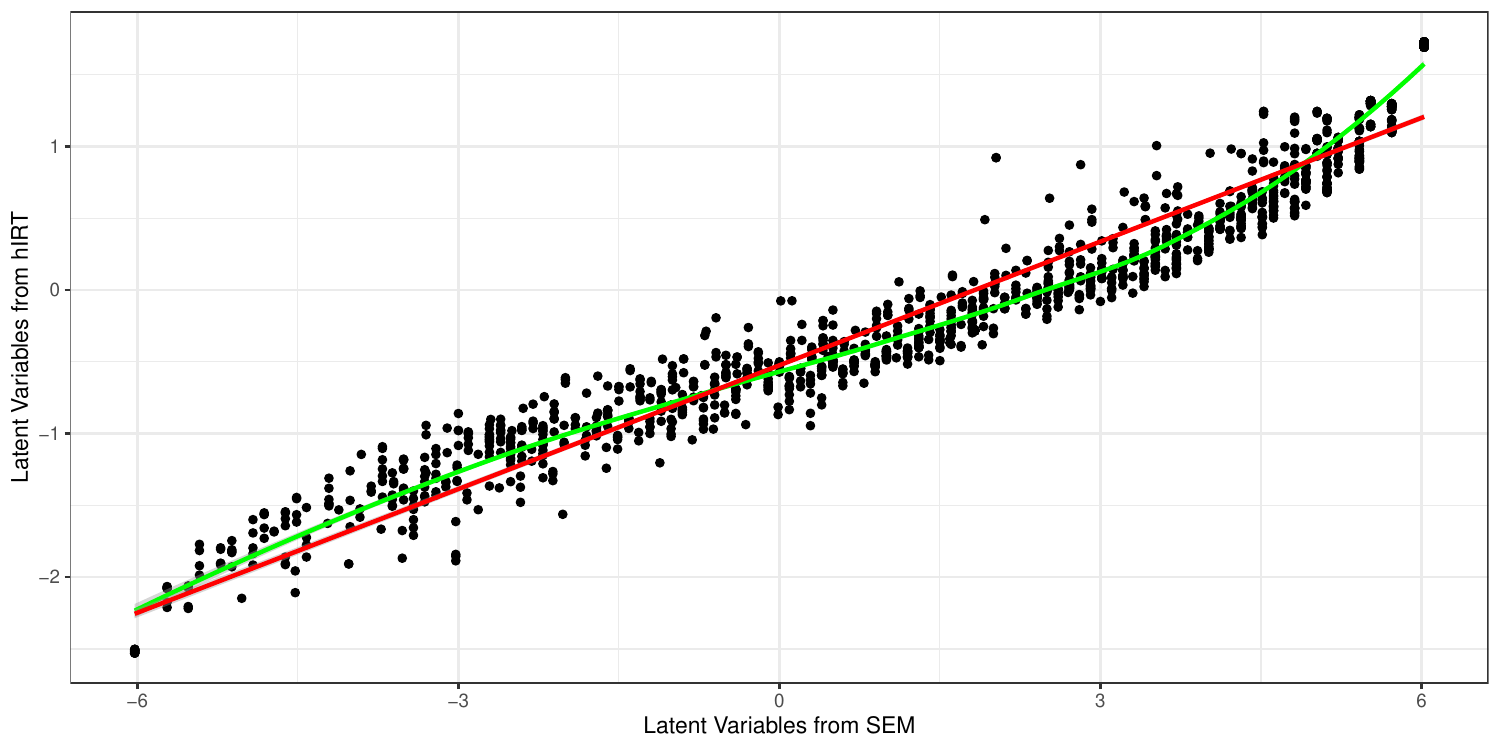}\caption{Scatterplot of the Imputed Latent Variables implied by the SEM and hIRT Models, with Linear and LOESS Fitted Lines}\label{fig:sem_hirt}
\end{figure}

\subsection{Linearity Tests and Robustness}

\clearpage

\section{Additive Indices Help Linearize the Relationship between the Observed and Latent Outcomes}\label{se:linearize}

Intuitively, we can think of each (discrete) measurement $Y_{ij}$ as being drawn from some distribution indexed by a latent variable $\eta_i$. More generally, suppose that the expectation of this distribution is given by $L_j(\eta_i,\alpha_j)$, where $L$ is a linear transformation, and $\alpha_j$ can be a vector of parameters (item difficulty for example). Then, if the correlations between measurement errors are not too strong, the sum of these measurements will converge to $\sum_{j=1}^J L_j(\eta_i,\alpha_j)$. As a result, we would expect the additive index to have a linear relationship with $\eta_i$. This formulation encompasses many common data-generating processes, for example, binary measurements (e.g., agree/disagree) drawn from a Bernoulli distribution or ordinal responses from an ordered probit model, among others. 

\begin{proposition}\label{prop:lln}
For each $i$, let $\{Y_{ij}\}$ be random variables with finite means $L_j(\eta_i,\alpha_j)<\infty$ and $\sup_j Var(Y_{ij})<\infty$. Define $\overline{Y}_i = \frac{1}{J} \sum_{i=1}^J Y_{ij}$. 

If  $\sum_{J=1}^\infty \frac{\sum_{j=1}^J\sum_{k=1}^J Cov(Y_{ij},Y_{ik})}{J^2}<\infty$, then $\overline{Y}_i$ converges to a linear function of $\eta_i$, 
given that $\lim_{J \rightarrow \infty} \frac{1}{J} \sum_{j=1}^J L_j(\eta_i,\alpha_j)$ exists.
\end{proposition}

\begin{proof}

Define Centered variables $X_{ij}=Y_{ij}-L_j(\eta_i,\alpha_j)$, so $\mathbb{E}X_{ij}=0$. Let $S_J=\sum_{j=1}^J X_{ij}.$ Therefore, the condition that $\sum_{J=1}^\infty \frac{\sum_{j=1}^J\sum_{k=1}^J Cov(Y_{ij},Y_{ijk})}{J^2}<\infty$ is equivalent to $\sum_{J=1}^\infty \frac{Var(S_J)}{J^2}< \infty$

By Chebyshev’s inequality, for any $\epsilon>0$, $\mathbb{P}[|S_J|\ge \epsilon J] \leq \frac{Var(S_J)}{\epsilon^2 J^2}$. Hence, $\sum_{J=1}^\infty \mathbb{P}[|S_J|\ge \epsilon J] \leq \sum_{J=1}^\infty \frac{Var(S_J)}{\epsilon^2 J^2}$. By assumption of $\sum_{J=1}^\infty \frac{Var(S_J)}{J^2}< \infty$, we confirm $\sum_{J=1}^\infty \mathbb{P}[|S_J|\ge \epsilon J]<\infty$. Next, by Borel-Cantelli lemma, we conclude $\lim_{J \rightarrow \infty}\frac{S_J}{J}=0$ a.s., which implies $\overline{Y}_i \rightarrow_{a.s.} \lim_{J \rightarrow \infty} \frac{1}{J} \sum_{j=1}^J L_j(\eta_i,\alpha_j)$. If the limit is well-define, it is still a linear function of $\eta_i$.
\end{proof}

If the measures for each individual were independent conditional on their latent score, then the above results reduce to Kolmogorov’s Strong Law of Large Numbers. Measures generated by the IRT model suggested by \citet{stoetzer2022causal} generally satisfy the above proposition as well. In sum, our framework offers a more flexible method to estimate latent treatment effects without assuming constant treatment effects or untestable nonlinear modeling assumptions.

\begin{figure}[!h]
    \centering
    \includegraphics[width=.9\linewidth]{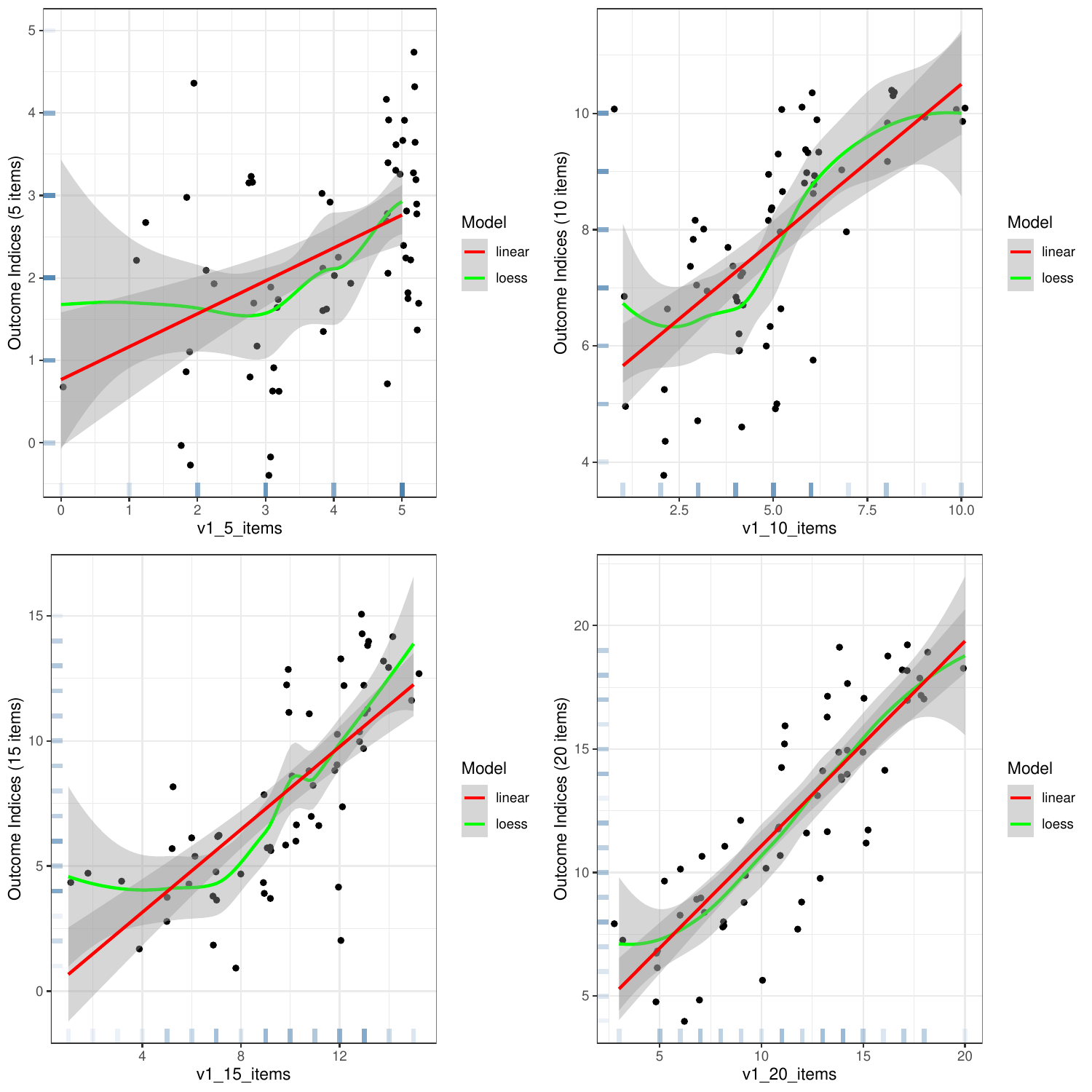}
    \caption{\textbf{Linearity between measurements by adding binary items together}. In each figure, we use the same data from figure \ref{fig:lin1}. We create two variables $v_1$ and $v_2$ by summing up 5, 10, 15, and 20 binary responses from IRT models. Data points are jittered. The rug plots on both axes denote the distribution of data.}
    \label{fig:lin2}
\end{figure}

\begin{figure}[!h]
    \centering
    \includegraphics[width=.9\linewidth]{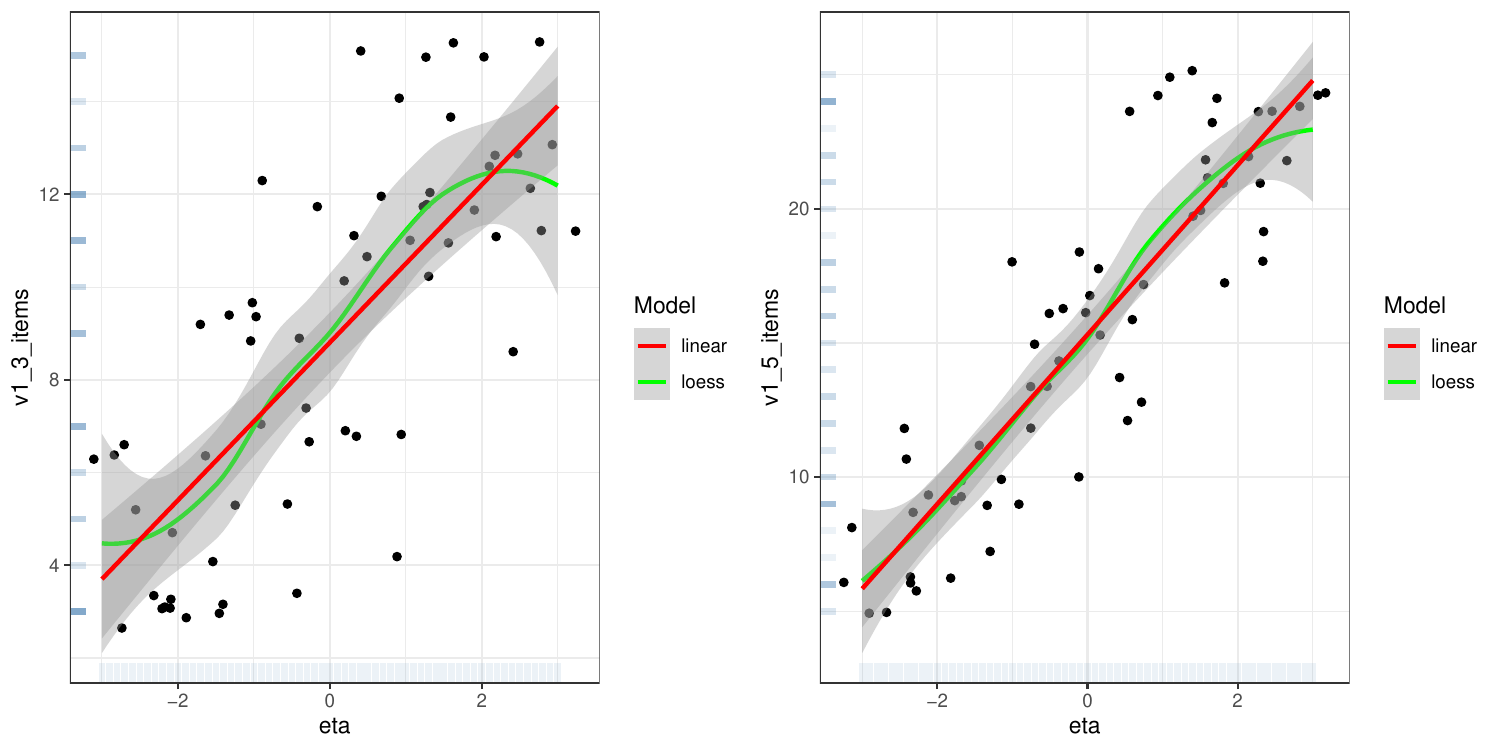}
    \caption{\textbf{Linearity between $\eta$ and measurements by adding 5-point ordinal items together}. We create an additive index $v_1$  by summing up 3 and 5 ordinal responses (taking values 1, 2, 3, 4, and 5) generated from IRT models. Data points have been jittered slightly for clarity. The rug plots on both axes denote the distribution of data. The vertical axis shows the additive index created by adding all ordinal variables in the simulation. The horizontal axis is the true latent variable ($\eta$) used in the IRT model.  The figures show that even a relatively small number of 5-point scales produce an index that bears an approximately linear relationship to the latent variable.}
    \label{fig:lin3}
\end{figure}

\section{A Worked Example of Parameter Identification Given a Linear Latent Variable Model}\label{si:struct}

This section shows how a research may go about establishing whether the parameters of a latent variable model are identified given an experimental design and maintained assumptions. Our example focuses on Figure \ref{fig:dgp}, which implies the following system of equations:

$$
\begin{aligned}
    \eta_i &= \beta Z_i + \zeta_i\\
    Y_{i1}&=\eta_i+\epsilon_{i1} =\beta Z_i+\zeta_i+\epsilon_{i1}\\
    Y_{i2}&=\lambda_2 \eta_i+\epsilon_{i2}=\lambda_2 \beta Z_i+\lambda_2 \zeta_i+\epsilon_{i2}\\
     Y_{i3}&=\lambda_3 \eta_i+\epsilon_{i3}=\lambda_3 \beta Z_i+\lambda_3 \zeta_i+\epsilon_{i3}
\end{aligned}
$$
Notice that $\lambda_1$ has been set to 1.0 in order to set the metric for the latent variable $\eta$.

Now, consider the population covariance matrix of the observed variables $Z,Y_1,Y_2,Y_3$,

$$
\begin{bmatrix}
    Var(Z_i) &  & & \\
    Cov(Y_{i1}, Z_i) & Var(Y_{i1}) & & \\
     Cov(Y_{i2}, Z_i) &  Cov(Y_{i2}, Y_{i1}) & Var(Y_{i2}) & \\
     Cov(Y_{i3}, Z_i) &  Cov(Y_{i3}, Y_{i1}) & Cov(Y_{i3}, Y_{i2})& Var(Y_{i3})
\end{bmatrix}
$$

Suppose our experiment were large enough to approximate the population covariance matrix. The identification problem may be stated as follows: if we knew the population covariance matrix, could we use the elements to solve for some or all of the model parameters? 

Each unique element of the covariance matrix implies one equation. On the left hand side of each equation is the estimated value of each variance or covariance; on the right hand side is the quantity implied by the model. In this example, we have 10 equations and 8 unknown parameters.

$$
\begin{aligned}
  \widehat{Var}(Z_i) &=Var(Z_i) \\
  \widehat{Cov}(Y_{i1}, Z_i)&=\beta Var(Z_i) \\
  \widehat{Cov}(Y_{i2}, Z_i)&=\lambda_2 \beta Var(Z_i) \\
   \widehat{Cov}(Y_{i3}, Z_i)&=\lambda_3 \beta Var(Z_i) \\
   \widehat{Cov}(Y_{i1}, Y_{i2}) &= \lambda_2 \beta^2  Var(Z_i) + \lambda_2 Var(\zeta_i) \\
   \widehat{Cov}(Y_{i1}, Y_{i3}) &= \lambda_3 \beta^2  Var(Z_i) + \lambda_3 Var(\zeta_i) \\
   \widehat{Cov}(Y_{i2}, Y_{i3}) &= \lambda_2 \lambda_3 \beta^2  Var(Z_i) + \lambda_2 \lambda_3 Var(\zeta_i)\\
   \widehat{Var}(Y_{i1}) &= \beta^2 Var(Z_i) + Var(\zeta_i) + Var(\epsilon_{i1})\\
   \widehat{Var}(Y_{i2}) &= \lambda_2^2 \beta^2 Var(Z_i) + \lambda_2^2 Var(\zeta_i) + Var(\epsilon_{i2})\\
    \widehat{Var}(Y_{i3}) &=  \lambda_3^2 \beta^2 Var(Z_i) + \lambda_3^2 Var(\zeta_i) + Var(\epsilon_{i1})
\end{aligned}
$$
In this case, all of the parameters are identified. We may establish the identification of each paramater by showing that it can be expressed as a function of the observed variances and covariances.  For example, the parameter $\beta$ may be solved by dividing  $\widehat{Cov}(Y_{i1}, Z_i)$ by $\widehat{Var}(Z_i)$. Certain parameters may be solved in more than one way, which implies that they are over-identified.

\section{Robustness}
\subsection{Exclusion Restriction Violations Stemming from Invalid Measurement} \label{si:more}

Consider the DGP in Figure \ref{fig:fake}. $Y_{i2}$ is affected by treatment $Z_i$ through other channels beyond the latent variable of interest $\eta_i$. Setting $\lambda_1=1$ to set the scale for $\eta_i$, we write 
$$
\begin{aligned}
    Y_{i1}&= (\eta_i^0+Z_i \tau_i) + \epsilon_{i1} \\
    Y_{i2}&= \lambda_2 (\eta_i^0+Z_i \tau_i)+\tilde{\lambda}_2 (\tilde{\eta}_i^0+Z_i \tilde{\tau}_i) + \epsilon_{i2} \\
    &:= \lambda_2 (\eta_i^0+Z_i \tau_i)+\tilde{\varepsilon}_{i2}
\end{aligned}
$$ Because $\tilde{\varepsilon}_{i2}$ is correlated with $Z_i$, assumption \ref{ass:frame}C is violated. When $\tilde{\lambda}_2 \neq 0$, $\lambda_2$ is not identified. The root of the problem is that $Z_i$ has a causal path to $Y_{i2}$ other than through $\eta_i$.  The covariance that $Y_{i1}$ and $Y_{i2}$ share reflects something other than their shared dependence on $\eta_i$; this covariance is affected also by the fact that $Z_i$ affects both $\eta_i$ and $\tilde{\eta}_i$.

This case illustrates a potential downside of misspecifying a latent variable model. Although a simple regression of $Y_{i1}$ on $Z_i$ would recover a meaningful causal effect, a latent variable model that ignores the backdoor path from $Z_i$ to $Y_{i2}$ through $\tilde{\eta}_i$ will produced biased estimates of the ALTE of $Z_i$ on $\eta_i$.  Under what conditions would this kind of misspecification arise?  One scenario occurs when $Y_{i2}$ is an invalid measure of $\eta_i$ insofar as it measures another latent trait as well, and this other trait is itself affected by the treatment $Z_i$.  Another scenario occurs when the treatment triggers a response bias that affects some measures but not others. This kind of artifact might occur if the treatment were to affect both a trait (e.g., authoritarian attitudes) as well as a source of mismeasurement (e.g., acquiescence to agree/disagree questions).



\begin{figure}[!h]
    \centering
    \includegraphics[width=0.8\linewidth]{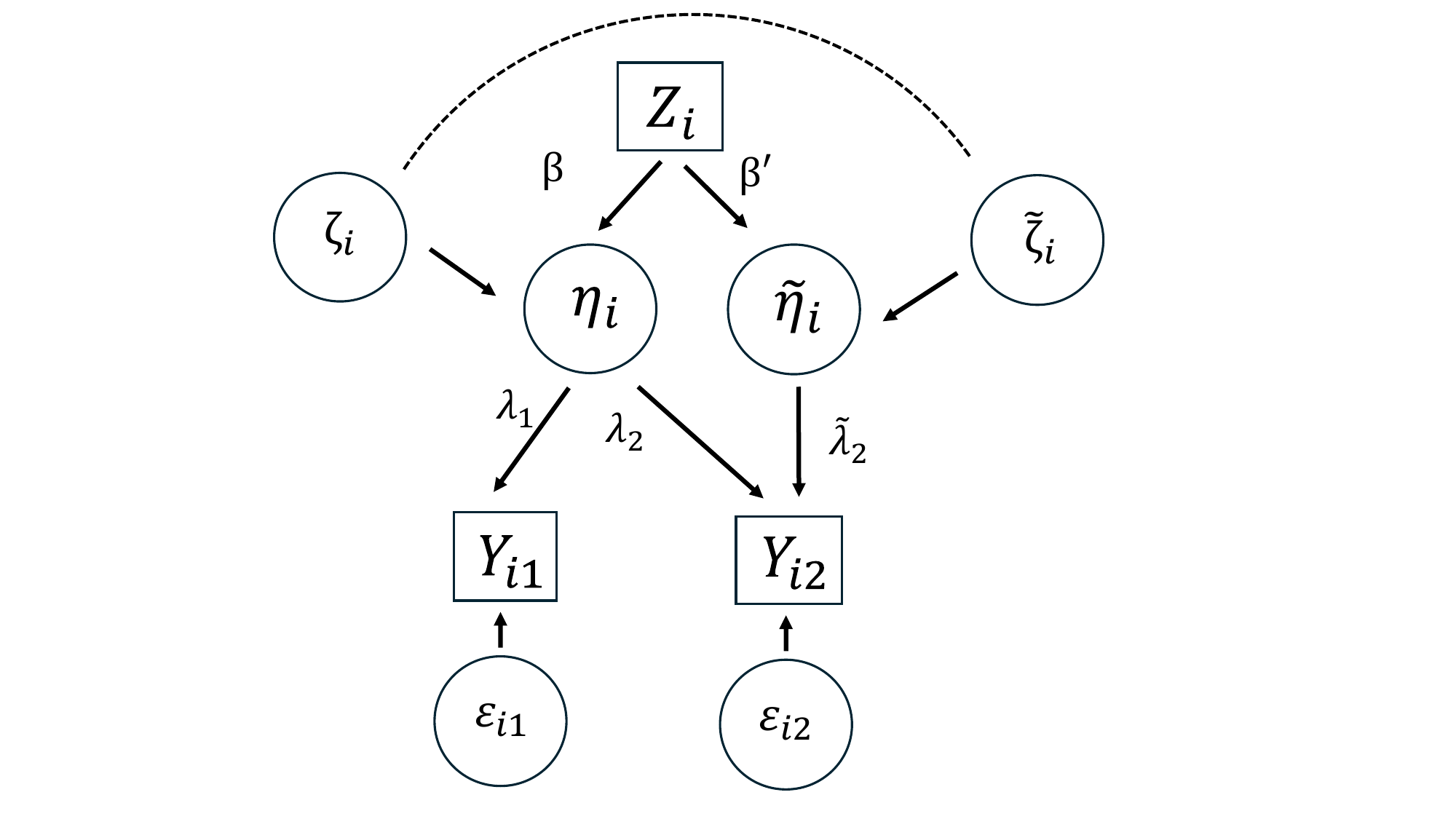}
    \caption{Data generating process violating the exclusion restriction.}
    \label{fig:fake}
\end{figure}

If so, how can we mitigate the problem? One solution is to collect additional valid measures. We would expect the bias from the invalid measure to be diluted if there are more valid outcome variables. In the simulation (1000 times), $\eta^0 = 0$ and $\eta^1 \sim N(0,1)$. We create valid outcomes $Y_{ij}=Z_i \eta^1 + (1-Z_1)\eta^0 + N(0,1)$ and an invalid measure $Y_{ij}=Z_i \eta^1 + (1-Z_1)\eta^0 + N(0,1) + 0.05Z_i$. The result in Figure \ref{fig:vio} confirms that when the number of valid measures increases, the bias and the variance decrease.  Another approach is to both collect more valid measurements and stipulate a more agnostic measurement model that allows some of the outcome measures to tap into more than one latent trait.

\begin{figure}[h]
    \centering
    \includegraphics[width=0.6\linewidth]{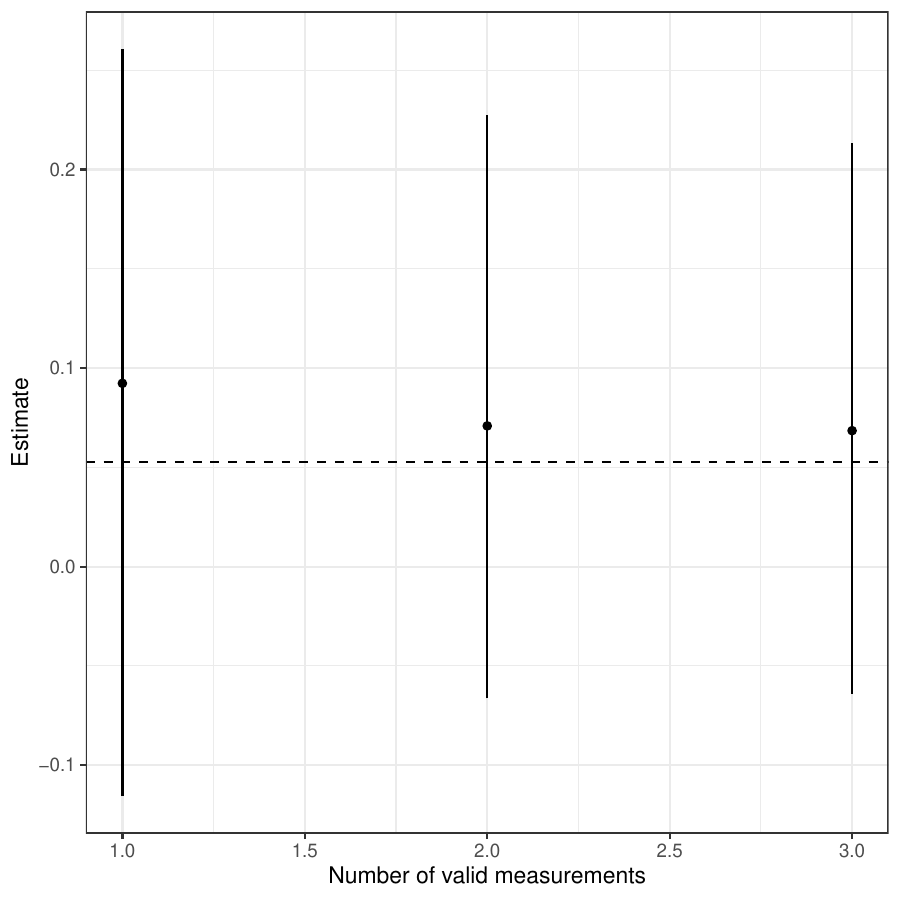}
    \caption{Simulation: Given an Invalid Outcome Measure, Adding Valid Outcome Measures Decreases the Bias and Variance of the Estimated ALTE}
    \label{fig:vio}
\end{figure}

\subsection{Robustness to Linearity Violations}\label{si:more1}

Suppose one measure is believed to violate the linearity assumption. Although linearity is an assumption of our modeling framework, a single nonlinear measure may not substantially bias the estimates produced by the WSI procedure so long as the analyst relaxes the linearity assumption for the suspect measure. 

In the following simulations, we examine the performance of WSI under quadratic and exponential nonlinearities. We set $\eta^0=0$ and $\eta^1 \sim N(1,1)$. The problematic nonlinear measure is: $Y_{i2}= Z_i [(\eta_i^1)^2] + (1-Z_i) [ (\eta_i^0)^2]+N(0,1)$, or alternatively, $Y_{i2}= Z_i [\exp(\eta_i^1)] + (1-Z_i) [ \exp(\eta_i^0)]+N(0,1)$, while the remaining measures $Y_{i1}$ and $Y_{i3}$ are linear: $Y_{ij}=Z_i \eta^1 + (1-Z_1)\eta^0 + N(0,1)$, $j=1,3$. 

We next consider the sampling distributions of two estimators. The first estimator incorrectly assumes that all measures are linear. The second estimator correctly assumes that $Y_{i2}$ is nonlinear.
See Figure~\ref{fig:nonlinear1} and \ref{fig:nonlinear2}.

\begin{figure}[!ht]
    \centering
    \includegraphics[width=0.7\linewidth]{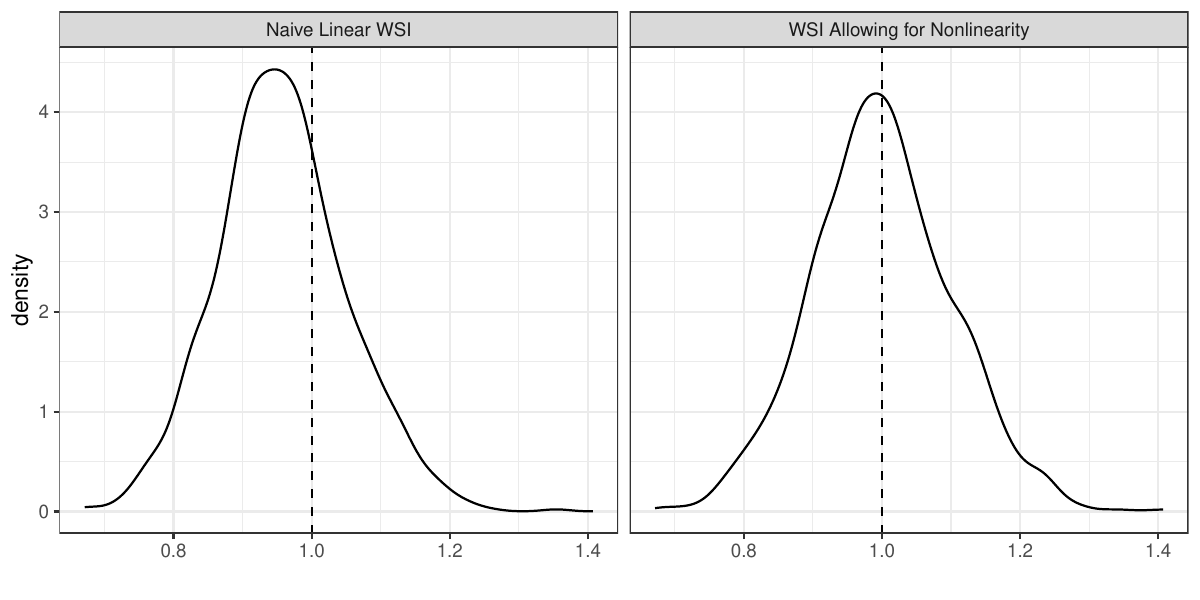}
    \caption{Robustness of Linear WSI and WSI Allowing for Nonlinearity in $Y_2$  under Exponential Nonlinearity.}
    \label{fig:nonlinear1}
\end{figure}

\begin{figure}[!ht]
    \centering
    \includegraphics[width=0.7\linewidth]{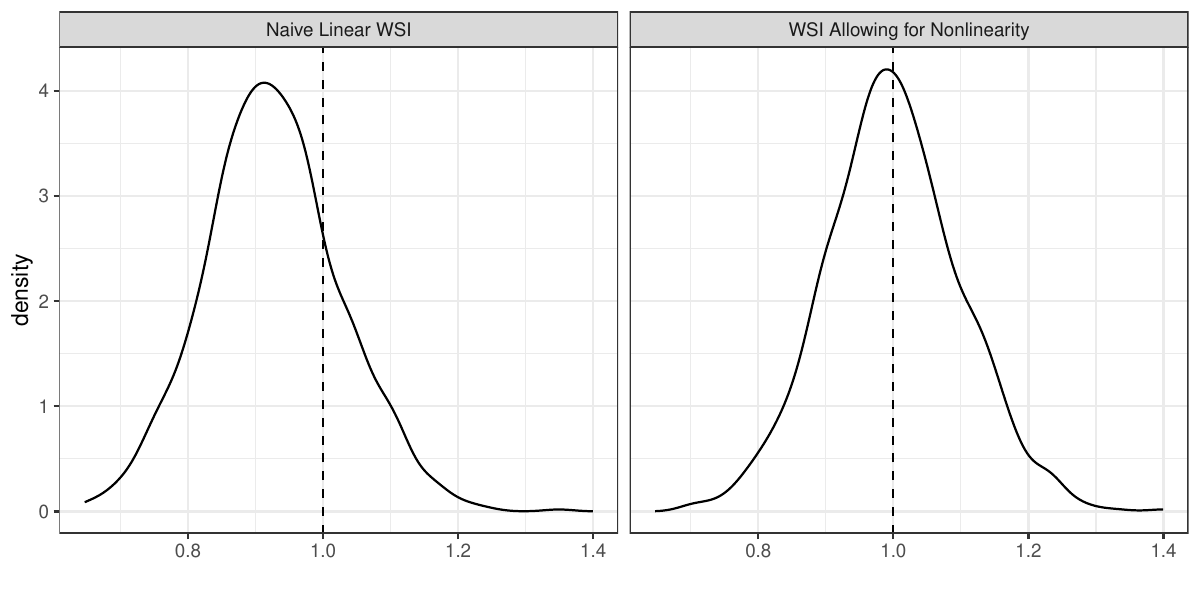}
    \caption{Robustness of Linear WSI and WSI Allowing for Nonlinearity in $Y_2$  under Quadratic Nonlinearity}
    \label{fig:nonlinear2}
\end{figure}

We find that the misspecified linear model is biased. The key reason for bias is that when $Y_{i2}$ is nonlinear, $Y_{i3}$ is no longer a valid instrument for identifying $\lambda_2$. The naive linear model assumes the functional form is correctly specified and continues to use $Y_{i3}$ as an instrument, which leads to significant bias. In contrast, WSI can still consistently estimate a linear approximation by using only $Z$ as the instrument. With valid linear information from $Y_{i1}$ and $Y_{i3}$, the WSI estimator is therefore more robust to the nonlinearity introduced by $Y_{i2}$.

To be more specific, because any nonlinear relationship can be approximated by polynomials, a researcher may specify a polynomial model linking $Y_{i2}$ to the latent outcome $\eta$, for example, $Y_{i2}=\sum_{k=1}^K \lambda_{2k} \eta_{i} + \epsilon_{i2}$. Then, following Lemma~\ref{lem}, we can use $\eta_{i1}=Y_{i1}-\epsilon_{i1}$ to replace for $\eta_i$ in the polynomial terms. One can verify that the remaining measures $Y_{ij}$ ($j \ne 1,2$) re not valid instruments. Consequently, we can only use the treatment assignment and its polynomial terms as instruments. Because $Z$ is binary in our setting, we effectively have only one instrument, $Z$. As a result, we can only identify the simplest linear specification, $Y_{i2}= \lambda_{2k} \eta_{i}+\epsilon_{i2}$. As long as this linear approximation is sufficiently accurate, it can still yield reasonable estimates. The caveat, however, is that although this linear approximation approach can, in principle, approximate a nonlinear relationship well, it differs from nonparametric methods in that it ignores additional estimation uncertainty arising from model misspecification.

The implication is that when nonlinearity is suspected in some measures but not others, robust estimation remains possible so long as the WSI estimator handles the linear and nonlinear measures differently when identifying and estimating the scaling parameters.

\clearpage








\putbib
\end{bibunit}
\end{document}